\documentclass[
notitlepage,%
floatfix,%
aps,%
prx,%
reprint,%
onecolumn,%
superscriptaddress,%
]{revtex4-2}

\usepackage{mathrsfs} 
\usepackage{bbm} 
\usepackage{bm}  
\usepackage{times} 
\usepackage[dvipsnames]{xcolor,colortbl}
\usepackage{mathtools}
\usepackage{amssymb} 
\usepackage{amsmath}
\allowdisplaybreaks
\usepackage[titletoc,toc,title]{appendix}
\usepackage{float}
\usepackage[colorlinks=true,
      linkcolor=blue,
      bookmarks=true,
      breaklinks=true,
      filecolor=blue,
      anchorcolor=yellow,
      citecolor=blue,
      urlcolor=blue,
      pdfauthor={Masahito Hyashi and Kun Wang},
      pdfsubject={Dense Coding with Locality Restriction for Decoder},
      CJKbookmarks=true]
      {hyperref}
\usepackage[shortlabels]{enumitem}
\usepackage{tikz}
\usetikzlibrary{shadows.blur}
\usetikzlibrary{shapes.symbols}
\usetikzlibrary{intersections,
                positioning,
                shapes,
                arrows}
\usepackage{multirow}
\usepackage{booktabs} 
\usepackage{theorem}
\usepackage{pifont}
\newtheorem{definition}{Definition}
\newtheorem{proposition}[definition]{Proposition}
\newtheorem{lemma}[definition]{Lemma}
\newtheorem{theorem}[definition]{Theorem}
\newtheorem{corollary}[definition]{Corollary}
\newtheorem{assumption}{Assumption}
\newenvironment{proof}{\noindent \textbf{{Proof~}}}{\hfill $\blacksquare$}
\newenvironment{proofwithpara}[1][]{\noindent\textbf{{{#1}~}}}{\hfill $\blacksquare$}

\newcommand\MH[1]{{#1}}

\newcounter{remark}
\newenvironment{remark}[1][]{\refstepcounter{remark}\par\medskip\noindent%
\textbf{Remark~\theremark #1} }{\medskip}

\mathchardef\ordinarycolon\mathcode`\:
\mathcode`\:=\string"8000
\def\vcentcolon{\mathrel{\mathop\ordinarycolon}}
\begingroup \catcode`\:=\active
  \lowercase{\endgroup
  \let :\vcentcolon
  }

\newcommand{\Renyi}{R\'{e}nyi~}
\newcommand{\wt}[1]{\widetilde{#1}}
\newcommand{\wh}[1]{\widehat{#1}}
\renewcommand{\ol}[1]{\overline{#1}}

\newcommand{\ket}[1]{\left\vert #1 \right\rangle}
\newcommand{\bra}[1]{\left\langle #1 \right\vert}
\newcommand{\braket}[2]{\langle#1\vert#2\rangle}
\newcommand\proj[1]{\vert #1 \rangle\!\langle #1 \vert}
\newcommand{\linear}[1]{\mathscr{L}(#1)}
\newcommand{\pos}[1]{\mathscr{P}(#1)}
\newcommand{\density}[1]{\mathscr{D}(#1)}
\newcommand{\subdensity}[1]{\mathscr{D}_{\bullet}(#1)}  
\newcommand{\unitary}[1]{\mathscr{U}(#1)}
\newcommand{\channel}[1]{\mathscr{C}(#1)}
\newcommand{\ox}{\otimes}
\newcommand{\1}{\mathbbm{1}}
\DeclareMathOperator{\tr}{Tr}  
\DeclareMathOperator{\rank}{rank}  
\newcommand{\id}{\operatorname{id}}
\newcommand{\rel}{\middle\|}

\newcommand{\sbar}{\;\rule{0pt}{9.5pt}\middle|\;}

\DeclareMathOperator{\Shannon}{H}
\DeclareMathOperator{\Mutual}{I} 
\DeclareMathOperator{\Rel}{D} 
\newcommand{\SRRel}[1]{\ensuremath{\widetilde{\operatorname{D}}_{#1}}}
\newcommand{\PRRel}[1]{\ensuremath{\overline{\operatorname{D}}_{#1}}}
\newcommand{\PRenyi}[1]{\ensuremath{\overline{\operatorname{H}}_{#1}}}

\newcommand*{\cB}{\mathcal{B}}
\newcommand*{\cC}{\mathcal{C}}
\newcommand*{\cD}{\mathcal{D}}
\newcommand*{\cE}{\mathcal{E}}

\newcommand*{\cG}{\mathcal{G}}
\newcommand*{\cH}{\mathcal{H}}

\newcommand*{\cK}{\mathcal{K}}
\newcommand*{\cM}{\mathcal{M}}
\newcommand*{\cN}{\mathcal{N}}

\newcommand*{\cS}{\mathcal{S}}
\newcommand*{\cT}{\mathcal{T}}
\newcommand*{\cU}{\mathcal{U}}
\newcommand*{\cX}{\mathcal{X}}
\newcommand*{\bE}{\mathbb{E}}
\newcommand*{\fD}{\mathfrak{D}}
\newcommand*{\fE}{\mathfrak{E}}

\newcommand{\fs}{\mathscr{F}} 
\def\Label#1{\label{#1}\ [\ \text{#1}\ ]\ }
\def\Label{\label}
\def\SEP{\mathrm{SEP}}
\def\PPT{\mathrm{PPT}}

\newcommand*{\EncG}{\ensuremath{\fE_{\rm g}}}
\newcommand*{\EncCP}{\ensuremath{\fE_{\rm cp}}}
\newcommand*{\EncP}{\ensuremath{\fE_{\rm p}}}
\newcommand*{\EncPPT}{\ensuremath{\fE_{\rm ppt}}}

\newcommand*{\DecL}{\ensuremath{\fD_{\emptyset}}}
\newcommand*{\DecOne}{\ensuremath{\fD_{\rightarrow}}}
\newcommand*{\DecLOCC}{\ensuremath{\fD_{\leftrightarrow}}}
\newcommand*{\DecSEP}{\ensuremath{\fD_{\rm sep}}}
\newcommand*{\DecPPT}{\ensuremath{\fD_{\rm ppt}}}
\newcommand*{\DecG}{\ensuremath{\fD_{\rm g}}}

\begin{document}

\title{Dense Coding with Locality Restriction for Decoder:\\%
Quantum Encoders vs. Super-Quantum Encoders}

\author{Masahito Hayashi}
\email{hayashi@sustech.edu.cn}
\affiliation{Shenzhen Institute for Quantum Science and Engineering,
Southern University of Science and Technology, Shenzhen 518055, China}
\affiliation{Guangdong Provincial Key Laboratory of Quantum Science and Engineering,
Southern University of Science and Technology, Shenzhen 518055, China}
\affiliation{Graduate School of Mathematics, Nagoya University, Nagoya, 464-8602, Japan}
\author{Kun Wang}
\email{wangkun28@baidu.com}
\affiliation{Institute for Quantum Computing, Baidu Research, Beijing 100193, China}

\begin{abstract}
We investigate dense coding by imposing various locality restrictions to our decoder
by employing the resource theory of asymmetry framework.
In this task, the sender Alice and the receiver Bob share an entangled state.
She encodes the classical information into it using
\MH{a symmetric preserving encoder}
and sends the encoded state to Bob through a
noiseless quantum channel.
The decoder is limited to a measurement to satisfy a certain locality condition
on the bipartite system composed of the receiving system and the preshared entanglement half.
Our contributions are summarized as follows:
First, we derive an achievable transmission rate for this task dependently of
conditions of encoder and decoder.
Surprisingly, we show that the obtained rate cannot be improved even when
the decoder is relaxed to local measurements, two-way LOCCs, separable measurements,
or partial transpose positive (PPT) measurements for the bipartite system.
Moreover, depending on the class of allowed measurements with a locality condition,
we relax the class of encoding operations to \MH{super-quantum encoders}
in the framework of general probability theory (GPT).
That is, when our decoder is restricted to a separable measurement,
theoretically, a positive operation is allowed as an encoding operation.
Surprisingly, even under this type of \MH{super-quantum} relaxation,
the transmission rate cannot be improved.
This fact highlights the universal validity of our analysis beyond quantum theory.
\end{abstract}
\maketitle

\tableofcontents

\section{Introduction}
Recently, various studies addressed how quantum theory can be characterized
in the context of general probability theory (GPT), which is
a general framework of the pair of a state and a measurement.
Most of them discussed it in the context of state discrimination
by extending the class of quantum measurements to a larger class of measurement in the framework of GPT
\cite{janotta2013generalized,janotta2014generalized,yoshida2020perfect,arai2019perfect,plavala2020popescu,plavala2020popescu,popescu1994quantum,lee2015computation,aubrun2020universal,lami2018ultimate,richens2017entanglement,masanes2011derivation,mueller2013three,muller2012structure,dahlsten2012tsirelson,barnum2010entropy,short2010entropy,matsumoto2018information,bae2016structure,yoshida2020asymptotic}.
That is, the paper \cite{arai2019perfect}
clarified that a super-quantum measurement, i.e., a measurement in such an extended class,
can distinguish two non-orthogonal states when our measurement belongs to the dual cone of the cone of separable operators.
It was also shown that a similar phenomenon happens even when
the cone of our measurements is very close to the cone of quantum theory \cite{yoshida2020perfect}.
This clarifies how quantum theory can be characterized in the framework of GPT.
Also, the paper showed the possibility of unphysical state
under the framework of GPT by introducing the PR box \cite{PRbox}.

In fact, a positive map can generate an unphysical state from an entangled state
when it is not a completely positive map.
This is a reason why the set of positive maps is not considered as a class of quantum operations and it is described by the set of completely positive maps.
Hence, a positive map can be considered as a super-quantum operation when it is not
a completely positive map.
Therefore, it is a natural question
whether such a super-quantum operation enhances quantum information processing.
In fact, the difference between positive maps and completely positive maps
has not been studied in the context of quantum information processing.
That is, such a study clarifies what property is essential for the analysis on quantum information processing
in a broad setting including quantum theory.
From this aim, we address this question in the framework of GPT.
In this paper, we focus on the information transmission problem as a typical model of
quantum information processing.
When a super-quantum operation is allowed as an encoding operation
in addition to a conventional quantum operation,
the decoder is needed to be restricted to a smaller class of measurements
under the framework of GPT.

To consider the effect of an unphysical state generated by a super-quantum operation,
it is natural to consider the following situation;
the sender and the receiver share an entangled state $\ket{\Psi_{AF}}$
and the encoding operation is limited to an operation on the sender's system $\cH_A$
while a super-quantum operation is allowed as an encoder.
The decoder is restricted to a smaller class of measurements
on the composite system under suitable locality condition,
which is chosen under the framework of GPT.
Since this encoding scheme may generate an unphysical state,
this problem setting clarifies the power of a super-quantum encoder.

Dense coding has been proposed as an enhanced communication method
by decoding measurement across a bipartite system composed of
the system encoded by the sender and the other system held by the receiver \cite{bennett1992communication,hiroshima2001optimal,bowen2001classical,horodecki2001classical,winter2002scalable,bruss2004distributed,beran2008nonoptimality,horodecki2012quantum,datta2015second,laurenza2019dense,wakakuwa2020superdense}.
Its subsequent studies investigated the channel capacity only when the channel is noisy.
The above studies assume that
the unique receiver obtains the system encoded by the sender as well as the other entanglement half $\cH_F$.
However, there is a possibility that
it is difficult to jointly handle the two systems,
the entanglement half $\cH_F$ of the receiver's system and
the receiving system $\cH_B$ of the receiver's side.
In this case, it is natural from the practical viewpoint
to impose locality conditions to our decoder.
For example,
when the receiver classically and freely communicates
from the entanglement half $\cH_F$ of the receiver's system to
the receiving system of the receiver's side,
we adopt the condition of one-way local operation and classical communication (LOCC).
When the receiver classically and freely communicates
between the two systems,
the condition of two-way LOCC is suitable.
As larger classes of decoders, we can consider separable measurements,
and partial transpose positive (PPT) measurements.
While their operational meaning are not so clear,
they have simple mathematical characterizations.
Hence, these classes are useful for proving the impossibility part.
Many existing papers addressed this problem in the contest of state discrimination
and state verification \cite{peres1991optimal,bennett1999quantum,walgate2000local,groisman2001nonlocal,virmani2001optimal,ghosh2001distinguishability,terhal2001hiding,watrous2005bipartite,hayashi2006bounds,hayashi2006study,owari2006local,koashi2007quantum,cohen2007local,owari2008two,ishida2008locality,matthews2009chernoff,duan2009distinguishability,hayashi2009group,jiang2009subspaces,calsamiglia2010local,jiang2010sufficient,bandyopadhyay2010entanglement,nathanson2010testing,kleinmann2011asymptotically,li2014relative,chitambar2013revisiting,childs2013framework,fu2014asymptotic,Brandao2014adversarial,owari2014asymptotic,owari2015local,hayashi2017tight,pallister2018optimal,zhu2019optimal,li2019efficient,wang2019optimal,yu2019optimal,hayashi2015verifiable,fujii2017verifiable,hayashi2018self,markham2020simple,zhu2018efficient,hayashi2019verifying,liu2019efficient,zhu2019efficient}.
However, no paper addressed this problem in the context of channel coding.

Indeed, the encoding operation in dense coding is given as an application of unitary
on Sender's system $\cH_A$.
However, in practice, it is not easy to implement all of the unitaries on the sender's system.
The time evolution in the quantum system is given as the unitary $e^{it H}$ with the Hamiltonian $H$.
Hence, this type of one-parameter subgroup $\{e^{it H}\}$ can be easily implemented.
When several types of Hamiltonian can be implemented,
the application of the subgroup generated by them can be implemented.
In this way,
it is natural to restrict our encoding operations into a certain subgroup.
That is, we assume that the encoding operation is given as a (projective) unitary representation $U$ of a group $G$ on $\cH_A$.
When the preshared entangled state is written as $\ket{\Psi_{AF}}$
and our coding operation is restricted to these unitaries,
our channel can be written as the classical-quantum (cq) channel
$g \mapsto U_g \ket{\Psi_{AF}}$.
Since this cq channel has a symmetric property for the group $G$,
we say that it is a {\it cq-symmetric} channel.
Recently, the paper \cite{korzekwa2019encoding} studied such a channel model in the context of resource theory of symmetry
without considering shared entanglement.

The class of cq-symmetric channels is a quantum generalization of
a regular channel \cite{delsarte1982algebraic}, which is a useful class of
channels in classical information theory.
This class of classical channels is often called
generalized additive \cite[Section V]{hayashi2011exponential}
or conditional additive \cite[Section 4]{hayashi2011exponential}
and contains a class of additive channels as a subclass.
Such a channel appears even in wireless communication
by considering binary phase-shift keying (BPSK) modulations \cite[Section 4.3]{hayashi2020finite}.
Its most simple example is the binary symmetric channel (BSC).
The reference \cite[Section VII-A-2]{hayashi2015quantum} studied its quantum extension with an additive group,
and discussed
the capacity and the wire-tap capacity with the semantic security.
Since this class has a good symmetric property,
algebraic codes achieve the capacity \cite{delsarte1982algebraic,dobrushin1963asymptotic,elias1955coding,hayashi2020finite,hayashi2015quantum}. Since a algebraic code has less calculation complexity in comparison with other types of codes,
this fact shows the usefulness of this class of classical channels.
As the class of cq-symmetric channels is a quantum version of a useful class of classical channels, and
the encoding operation class of group representation of a subgroup
yields a cq-symmetric channel,
this encoding operation class is a natural class of encoders.

To consider super-quantum encoders,
we need to expand the above class of encoding operations
because these unitatries are conventional quantum operations.
For this aim, we focus on a basic property of these unitaries.
These unitary encoding operations has invariant states, which
can be characterized by the invariant state by the average operation $\cG$ of the given (projective) unitary representation.
In the case of full unitary and the case of discrete Weyl-Heisenberg group,
the average operation $\cG$ maps all densities to the completely mixed state.
When the group is composed of diagonal unitaries,
the average operation $\cG$ is the dephasing channel.
Interestingly, the average operation $\cG$ satisfies the property
$\cG\circ \cU_g=\cU_g\circ \cG=\cG$ for $g \in G$, where
$\cU_g (\rho):=U_g \rho U_g^\dagger$.
Hence, as a larger class of encoding operations, we can consider the set of trace-preserving completely
positive (TPCP) maps $\{\cE\}$ that satisfies \MH{the symmetric preserving condition}:
\begin{align}
\cG\circ \cE=\cE\circ \cG=\cG.\Label{CO1}
\end{align}
Therefore, as another problem setting, we assume that our encoders are restricted to the above class of
TPCP maps.

Using this condition \eqref{CO1}, we define
a class of trace preserving positive maps as a larger class of encoders.
Indeed, when we focus on a basis commutative with invariant states,
the transpose operation satisfies the condition \eqref{CO1}.
Hence, this class of encoders contains a typical super-quantum operation.
In addition, since our measurement class is restricted,
this class is theoretically allowed as a class of encoding operations
under the framework of general probability theory (GPT).

Recently, several papers studied state discrimination in this framework, but
no study discussed the channel coding in this framework.
If our measurement class is smaller than the set of all measurements
allowed in quantum theory, a larger class of states is allowed, i.e.,
a larger class of operations is allowed theoretically in this framework.
The reason is that the probability distribution of the measurement outcome is well defined in this relaxation.
That is, the non-negativity of the probability of the measurement outcome is guaranteed even
under this relaxation.
For example,
when our measurement is restricted to a separable measurement,
even when the encoding operation is relaxed to a positive map,
the non-negativity of the probability of the measurement outcome is guaranteed
while the resultant state of the encoding is not necessarily positive-definite.
That is, the separability of our measurement guarantees
the non-negativity of the probability of the measurement outcome.
In this way, we can extend our encoding operation \MH{to such super-quantum operations}
under the condition \eqref{CO1} when a certain locality condition is imposed to our decoding measurement.

In this paper, we introduce $21$ classes of dense coding codes
by considering various classes of encoders and decoders.
These classes are classified into three groups dependently on the class of decoders.
The first group is composed of classes whose decoder has no support by $\cH_F$.
The second group is composed of classes whose decoder is a global measurement.
The remaining group is composed of
classes whose decoder has support from $\cH_F$ and locality condition on the bipartite system.
As our main result, we show that each class of every group has the same capacity.
That is, if a class belongs to the same group as another class,
these two classes have the same capacity.
Hence, when the available decoder is one of the bipartite decoders with locality condition,
even when the class of our encoders is extended to a larger class, e.g., trace-preserving positive maps
with the condition \eqref{CO1},
the capacity cannot be improved.

This paper is organized as follows.
Section \ref{S2} prepares several basic knowledge for this paper.
In Section~\ref{sec:task}, we first formally define the abstract dense coding task.
Then, we generalize the problem setting by
considering various available sets of encoders imposed by the resource theory of asymmetry
and various available sets of decoders with locality conditions.
As last, we show that all these capacities are equal and derive a single-letter capacity formula.
In Section~\ref{sec:examples}, we investigate various unitary groups of practical interests---the irreducible
case including the full unitary group,
quantum coherence including the one-generator case with a certain condition, e.g.,
two-mode squeezed vacuum state, and
Schur-duality---to illustrate the dense coding power within different specialized resource theories of asymmetry.
Finally, in Section \ref{S8} we extend the obtained results to the case of non-quantum preshared state
within the framework of GPT. We reserve some details and proofs to the appendices.
In Appendix~\ref{appx:thm:enhanced-relation},
we prove the weak and strong converse bounds on the (enhanced) dense coding capacities.
In Appendix~\ref{S4B}, we prove our main result---the dense coding capacity theorem under locality conditions.
We do so by first giving an one-shot characterization to the dense coding capacity with one-way LOCC decoders.
This is done by showing an achievability bound in terms of the smooth \Renyi entropy
and turns out to be the most difficult part in this paper.
Then, we derive the capacity formulas for the asymptotic dense coding capacities under the locality conditions.
In Appendix~\ref{S6}, we show prove the dense coding theorem with local decoders
even when super-quantum encoding operation is allowed.
In Appendix~\ref{appx:w-con2}, we prove the achievability (under certain conditions)
and strong converse parts regarding the non-quantum preshared state extension.

\section{Preliminaries}\Label{S2}

\subsection{Notations}
For a finite-dimensional Hilbert space $\cH$, we denote by $\linear{\cH}$ and $\pos{\cH}$ the linear and positive
semidefinite operators on $\cH$. Quantum states are in the set $\density{\cH}:=\{\rho\in\pos{\cH}\vert\tr\rho=1\}$ and
we also define the set of subnormalized quantum states $\subdensity{\cH}:=\{\rho\in\pos{\cH}\vert0<\tr\rho\leq1\}$. For
two operators $M, N\in\linear{\cH}$, we say $M\geq N$ if and only if $M-N\in\pos{\cH}$. On the other hand, we denote by
$\{M\geq N\}$ the projector onto the space spanned by the eigenvectors of $M-N$ with non-negative eigenvalues. The
identity matrix is denoted as $\1$ and the maximally mixed state is denoted as $\pi$. Multipartite quantum systems are
described by tensor product spaces. We use capital Latin letters to denote the different systems and subscripts to
indicate on what subspace an operator acts. For example, if $M_{AB}$ is an operator on $\cH_{AB}=\cH_A\ox\cH_B$, then
$M_A=\tr_BM_{AB}$ is defined as its marginal on system $A$. Systems with the same letter are assumed to be isomorphic:
$A'\cong A$. By convention, we use letters in the front of the alphabet such as $A$ and $B$ to represent quantum
systems and letters in the end of the alphabet such as $X$ and $Y$ to represent classical systems. We say
$\rho_{XA}$ is a classical-quantum state if it is of the form $\rho_{XA}=\sum_xp_X(x)\proj{x}_X\ox\rho_A^x$,
where $p_X$ a probability distribution, $\{\ket{x}\}_x$ an orthonormal basis of $\cH_X$,
and $\{\rho^x_A\in\density{\cH_A}\}_x$. A linear map
$\cN:\linear{\cH_A}\to\linear{\cH_B}$ maps operators in system $A$ to operators in system $B$. $\cN_ {A\to B}$ is
positive if $\cN_{A\to B}(M_A)\in\pos{\cH_B}$ whenever $M_A\in\pos{\cH_A}$. Let $\id_A$ denote the identity map
acting on system $A$. $\cN_{A\to B}$ is completely positive (CP) if the map $\id_R\ox\cN_{A\to B}$ is positive for
every reference system $R$. $\cN_{A\to B}$ is trace-preserving (TP) if $\tr[\cN_{A\to B}(M_A)] = \tr M_A$ for all
operators $M_A\in\linear{\cH_A}$. If $\cN_{A\to B}$ is completely positive and trace-preserving (CPTP), we say that it
is a quantum channel or quantum operation. We denote by $\channel{A\to B}$ the set of quantum channels from $A$ to $B$.
A positive operator-valued measure (POVM) is a set $\{\Lambda_m\}$ of operators satisfying $\forall m,\Lambda_m\geq0$
and $\sum_m\Lambda_m=\1$.

\subsection{Quantum entropies}\Label{sec:Quantum entropies}

Let $\rho\in\density{\cH}$ and $\sigma\in\pos{\cH}$ such that the support of $\rho$ is contained in the support of
$\sigma$. The quantum relative entropy is defined as
\begin{align}
  \Rel\left(\rho\rel\sigma\right) := \tr\left[\rho(\log\rho - \log\sigma)\right],
\end{align}
where logarithms are in base $2$ throughout this paper. The Shannon entropy of a probability distribution $p_X$ is
defined as $\Shannon(p_X):=-\sum_xp_X(x)\log p_X(x)$. The von Neumann entropy of $\rho$ is defined as
$\Shannon(\rho):=-\tr\rho\log\rho$.
Let $\rho_{AB}\in\density{\cH_A\ox\cH_B}$ be a bipartite quantum state.
The quantum mutual
information and conditional entropy of $\rho_{AB}$ are defined respectively as
\begin{align}
  \Mutual\left(A{:}B\right)_\rho &:= \Rel\left(\rho_{AB}\rel\rho_A\ox\rho_B\right), \\
  \Shannon\left(A{\vert}B\right)_\rho &:= -\Rel\left(\rho_{AB}\rel\1_A\ox\rho_B\right).
\end{align}
Trivializing system $B$, $\Shannon\left(A{\vert}B\right)_\rho$ yields an alternative definition
of the von Neumann entropy as $\Shannon(A)_\rho$. In this paper, we will use these two notations interchangeably.
Let $\alpha\in(0,1)\cup(1,\infty)$, the one-parameter Petz quantum \Renyi divergence is defined as~\cite{petz1986quasi}
(We refer the interested readers to~\cite[Chapter 4]{tomamichel2015quantum}
and \cite[Sections 3.1 and 5.4]{hayashi2016quantum} for a comprehensive study of $\PRRel{\alpha}$):
\begin{align}
    \PRRel{\alpha}\left(\rho\rel\sigma\right)
:=  \frac{1}{\alpha-1}\log\tr\left[\rho^\alpha\sigma^{1-\alpha}\right].
\end{align}
As another version, we focus on the one-parameter sandwiched quantum \Renyi divergence
$\SRRel{\alpha}$ defined as~\cite{wilde2014stronga,mueller-lennert2013quantum,frank2013monotonicity,beigi2013sandwiched}:
\begin{align}
    \SRRel{\alpha}\left(\rho\rel\sigma\right)
:=  \frac{1}{\alpha-1}\log\tr\left[
\left(
\sigma^{\frac{1-\alpha}{2\alpha}}
\rho
\sigma^{\frac{1-\alpha}{2\alpha}}
\right)^\alpha\right]. \Label{MO2}
\end{align}

Interestingly, both quantum \Renyi divergence recovers the quantum relative entropy by taking the limit $\alpha\to1$:
\begin{align}\Label{eq:Petz divergence limit}
    \lim_{\alpha\to1}\PRRel{\alpha}\left(\rho\rel\sigma\right) =
    \lim_{\alpha\to1}\SRRel{\alpha}\left(\rho\rel\sigma\right) = \Rel\left(\rho\rel\sigma\right).
\end{align}
The quantum \Renyi entropy of $\rho$ is defined as
\begin{align}\Label{eq:renyi entropy}
  \PRenyi{\alpha}(\rho) := - \PRRel{\alpha}\left(\rho\rel\1\right)
      = \frac{1}{1-\alpha}\log\tr\rho^\alpha.
\end{align}
Eq.~\eqref{eq:Petz divergence limit} yields $\lim_{\alpha\to1}\PRenyi{\alpha}(\rho)=\Shannon(\rho)$. The Petz
conditional \Renyi entropy of a bipartite state $\rho_{AB}$ is defined as
\begin{align}\Label{eq:renyi conditional}
      \PRenyi{\alpha}(A{\vert}B)_\rho := - \PRRel{\alpha}\left(\rho_{AB}\rel\1_A\ox\rho_B\right).
\end{align}

\subsection{Group representation}\Label{sec:group representation}
As a preparation to explain resource theory of asymmetry,
we summarize the basic fact in group representation.
Let $\cH$ be a Hilbert space and $G$ be a group.
For an element $g \in G$, a unitary operator $U_g$ is given.
The map $U: g \mapsto U_g$ is called
a unitary representation of $G$ on $\cH$ when
\begin{align}
U_{g} U_{g'}=U_{g g'}
\end{align}
for $g,g' \in G$~\cite{hayashi2017group}.
In addition, the map $U$ is called
a projective unitary representation of $G$ on $\cH$ when
there exists $\theta(g,g')$ for $g,g' \in G$ such that~\cite{hayashi2017group}
\begin{align}
U_{g} U_{g'}=e^{i \theta(g,g')}U_{g g'}.
\end{align}
In particular, the above type of
a projective unitary representation is called
a projective unitary representation associated with
$\{\theta(g,g')\}_{g,g' \in G}$.

A (projective) unitary representation $U$ of $G$ on $\cH$ is called irreducible
when there is no subspace $\cK \subsetneqq \cH$ such that
$U_g \cK \subset \cK$ for $g \in G$ and $\{0\}\neq \cK$.
Two (projective) unitary representations $U_1,U_2$ of $G$ on $\cH_1,\cH_2$
are called equivalent when
there exists a unitary $V$ from $\cH_1$ to $\cH_2$ such that
$ V U_{1,g}V^\dagger= U_{2,g}$ for $g \in G$.
A (projective) unitary representation $U$ of $G$ on $\cH$ is called
completely reducible when it is given as a direct sum representation of (projective) irreducible unitary representations.
It is known that
any unitary representation $U$ of $G$ on $\cH$ is
completely reducible when $G$ is a compact group \cite[Lemma 2.3]{hayashi2017group2}.

Let $\hat{G}$ be the set of indexes to identify an irreducible unitary representation of $G$.
That is, given an element $\lambda \in \hat{G}$, we have
an irreducible unitary representation $U_{\lambda}$ of $G$ on $\cH_\lambda$.
Then, any unitary representation $U$ of a compact $G$ on $\cH$ is equivalent to
\begin{align}
\bigoplus_{\lambda \in \hat{G}} \cH_{\lambda} \otimes \mathbb{C}^{n_\lambda} ,
\Label{M1}
\end{align}
where $n_\lambda$ is called the multiplicity of
the irreducible unitary representation $U_\lambda$.

\begin{assumption}[Multiplicity-free condition]\label{assp:multiplicity-free}
We say that a unitary representation $U$ of a $G$ on $\cH$ is
{\it multiplicity-free} when
there exists a subset $S \subset \hat{G}$
such that
the unitary representation $U$ is equivalent to
\begin{align}
\bigoplus_{\lambda \in S} \cH_{\lambda}.
\Label{M3}
\end{align}
Throughout this manuscript, we assume the groups under investigation satisfy the multiplicity-free condition.
\end{assumption}

The above discussion can be extended to projective unitary representations.
Let $\hat{G}[\{\theta(g,g')\}_{g,g' \in G}]$ be the set of indexes to identify an irreducible projective unitary representation
of $G$ associated with $\{\theta(g,g')\}_{g,g' \in G}$.
Then, Eq.~\eqref{M1} is generalized as follows;
Any unitary representation $U$ of a compact $G$ on $\cH$ associated with $\{\theta(g,g')\}_{g,g' \in G}$
is equivalent to
\begin{align}
\bigoplus_{\lambda \in \hat{G}[\{\theta(g,g')\}_{g,g' \in G}]} \cH_{\lambda} \otimes \mathbb{C}^{n_\lambda} ,
\Label{M2}
\end{align}
where $n_\lambda$ is called the multiplicity of
the irreducible unitary representation $U_\lambda$.
Hence, we define the property ``multiplicity-free'' for a projective unitary representation $U$
in the same way.

A state $\sigma \in \density{\cH}$ is \emph{symmetric} w.r.t. (with respect to) $G$ if
it holds that
\begin{align}\Label{eq:symmetric}
    \forall g \in G,\; \cU_g(\sigma) \equiv U_g \sigma U_g^\dag = \sigma.
\end{align}
That is, the symmetric states are invariant under $G$. Throughout this paper, we assume the group $G$ is fixed and omit
the explicit reference to this group. The set of symmetric states is denoted as $\fs_G$ and will be treated as
\emph{free states} in the resource theory of asymmetry. Conversely, a state $\rho \in
\density{\cH}$ is \emph{asymmetric}, or resourceful, if there exists some $g\in G$ such that $U_g\rho U_g^\dag
\neq\rho$.
When the group $G$ is a finite group,
the \emph{$G$-twirling operation} $\cG$ over $G$ is defined as
\begin{align}\Label{eq:twirling channel}
  \cG(\rho) := \frac{1}{\vert G\vert}\sum_g U_g\rho U_g^\dagger.
\end{align}
When the group $G$ is a compact group,
the above definition can be generalized as
\begin{align}\Label{eq:twirling channel2}
  \cG(\rho) := \int_G  U_g\rho U_g^\dagger \nu(dg),
\end{align}
where $\nu$ is the Haar measure.
$\cG$ maps all states in $\density{\cH}$ to
symmetric states, i.e.,
\begin{equation}
\forall\rho\in\density{\cH},\; \cG(\rho)\in\fs_G.
\end{equation}
What's more, $\cG$ is {symmetry-preserving} in the sense that it maps any symmetric state to itself:
$\cG(\sigma) = \sigma$ for any symmetric state $\sigma\in\fs_G$. One can interpret $\cG$ as a {resource
destroying map}~\cite{liu2017resource} in the sense that it leaves resource-free states unchanged but erases the
resource stored in all resourceful states.

\begin{lemma}
Let $G$ is a compact group.
A (projective) unitary representation $U$ is multiplicity-free if and only if
\begin{align}
\cG(\rho) \cG(\sigma)= \cG(\sigma)\cG(\rho)
\Label{LO1}
\end{align}
for two states $\rho,\sigma$ on $\cH$.
\end{lemma}

\begin{proof}
Assume that a (projective) unitary representation $U$ is multiplicity-free, as shown in \eqref{M3}.
Let $P_\lambda$ be the projection to the space $\cH_\lambda$.
Then,
$\cG(\rho)= \oplus_{\lambda \in S} \tr (P_\lambda \rho) P_\lambda/\tr P_\lambda$, which implies
\eqref{LO1}.
Assume that a (projective) unitary representation $U$ is not multiplicity-free.
Then, a state $\rho$ is written as
$\oplus_{\lambda} \rho_\lambda$, where
$\rho_\lambda$ is a  positive semidefinite operator on
$\cH_\lambda \otimes \mathbb{C}^{n \lambda} $.
Then,
$\cG(\rho)= \oplus_{\lambda \in S}
P_\lambda/\tr P_\lambda \otimes (\tr_{\cH_\lambda} \rho_\lambda) $.
For $n_\lambda >1$,
$\tr_{\cH_\lambda} \rho_\lambda$ and
$\tr_{\cH_\lambda} \sigma_\lambda$ are not commutative with each other in general.
Hence, we obtain the desired statement.
\end{proof}

The above definition generalizes naturally to a tensor product system
$\cH^{\ox n}$ composed of $n$ copies of $\cH$. The group in $\cH^{\ox n}$ is $G^{\times n}$ and we adopt the notations
$\bm{g}\equiv g_1\cdots g_n$ and $U_{\bm{g}}\equiv U_{g_1}\ox\cdots\ox U_{g_n}$ such that each $g_i\in G$. Symmetric
states are defined to be those satisfying $U_{\bm{g}}\sigma U_{\bm{g}} = \sigma$ for all $U_{\bm{g}}\in G^{\times n}$.
Correspondingly, the twirling operation in $\cH^{\ox n}$ is $\cG^{\ox n}$.


\subsection{Resource theory of asymmetry}\Label{sec:resource theory of asymmetry}
Asymmetry of quantum states plays an important role not only in the development of modern physics but also in quantum
information processing tasks~\cite{bartlett2007reference}. In this section, we summarize briefly the \textit{resource
theory of
asymmetry}~\cite{bartlett2007reference,gour2008resource,gour2009measuring,bartlett2009quantum,marvian2012symmetry,marvian2013theory,marvian2014asymmetry,marvian2014extending,marvian2014modes,wakakuwa2017symmetrizing,wakakuwa2020superdense},
which is a special case of a general formalism named the quantum resource theory~\cite{chitambar2018quantum}. We remark
that the resource theory of asymmetry is an \emph{abstract} resource theory that encapsulates many nice properties of
commonly studied resource theories in literature~\cite{gour2015resource,streltsov2017colloquium}.

The relative entropy of asymmetry~\cite{gour2009measuring} is a
commonly used measure to quantify the degree of asymmetry of quantum states and is defined as follows:
\begin{align}\Label{eq:relative entropy of asymmetry}
    R_G(\rho) := \min_{\sigma \in \fs_G} \Rel\left(\rho \rel \sigma\right).
\end{align}
It turns out that the twirled state $\cG(\rho)$ achieves the minimum in~\eqref{eq:relative entropy of asymmetry}
and yields a simple expression for the relative entropy of asymmetry
in terms of the von Neumann entropy~\cite[Proposition 2]{gour2009measuring}, i.e.,
\begin{align}\Label{eq:REA}
    R_G(\rho) = \Rel\left(\rho\rel\cG(\rho)\right) = \Shannon(\cG(\rho)) - \Shannon(\rho).
\end{align}

Inspired by the entanglement of assistance~\cite{divincenzo1998entanglement} in
the resource theory of entanglement~\cite{horodecki2009quantum} and the coherence of
assistance~\cite{chitambar2016assisted} in the resource theory of coherence~\cite{streltsov2017colloquium}, we
introduce here the \emph{asymmetry of assistance} of a quantum state $\rho$ as
\begin{align}
   A_G\left(\rho\right)
:=&\; \max_{\rho=\sum_xp_X(x)\proj{\psi_x}} \sum_x p_X(x)
      \Rel\left(\psi_x\rel\cG(\psi_x)\right) \nonumber\\
 =&\; \max_{\rho=\sum_xp_X(x)\proj{\psi_x}} \sum_x p_X(x)\Shannon\left(\cG(\psi_x)\right),
\Label{eq:asymmetry-of-assistance}
\end{align}
where $\psi_x\equiv\proj{\psi_x}$,
the maximum ranges over all possible pure state decompositions of $\rho$,
and the second equality follows from~\eqref{eq:REA}. Correspondingly, the \emph{regularized
asymmetry of assistance} of $\rho$ is defined as
\begin{align}\Label{eq:egularized-asymmetry-of-assistance}
    A_G^\infty\left(\rho\right) := \limsup_{n\to\infty}\frac{1}{n}A_G\left(\rho^{\ox n}\right).
\end{align}
In the following proposition, we show that both $A_G$ and $A_G^\infty$ are upper bounded by
the quantum entropy of the twirled state and thus the regularization is well-defined.
The proof can be found in Appendix~\ref{appx:prop:regularized-asymmetry-of-assistance}.

\begin{proposition}\Label{prop:regularized-asymmetry-of-assistance}
Let $\rho\in\density{\cH}$ be a quantum state. It holds that
\begin{align}\Label{eq:asymmetry-of-assistance-upper-bound}
  A_G\left(\rho\right) \leq A_G^\infty\left(\rho\right) \leq \Shannon(\cG(\rho)).
\end{align}
\end{proposition}

\section{Dense coding capacities}\Label{sec:task}

\subsection{The general dense coding framework}\label{sec:framework}

We first describe the most general dense coding framework.
Let the preshared entangled state between Alice and Fred be $\ket{\Psi}_{AF}$,
the set of available encoders by Alice be $\fE$,
and the set of available decoders by Bob and Fred be $\fD$.
The abstract dense coding protocol can be described as follows.
Alice randomly samples a message $m$ from the message alphabet $\cM$
and then applies an encoding channel $\cE^m_{A\to A}\in\fE$
to the resourceful state $\Psi_{AF}$.
This leads to the classical-quantum state
\begin{align}
  \frac{1}{\vert\cM\vert}\sum_m\proj{m}_M \ox\cE^m_{A\to A}(\proj{\Psi}_{AF}).
\end{align}
After encoding, Alice sends the encoded state to Bob via a noiseless quantum channel $\id_{A\to B}$ where $A\cong B$.
After receiving the quantum state, Bob and Fred perform a
joint measurement $\cD_{BF\to\wh{M}} \equiv \{\Gamma^{\wh{m}}_{BF}\}_{\wh{m}}\in\fD$
to infer the encoded message $m$. See Figure~\ref{fig:dense-coding} for illustration of the dense coding protocol.
The decoding operation results in the following classical-classical quantum state
\begin{align}
\sum_{m,\wh{m}}q_{\wh{M}M}(\wh{m}\vert m)\proj{m}_M \ox \proj{\wh{m}}_{\wh{M}},
\end{align}
where the conditional distribution $q_{\wh{M}M}$ is defined as
\begin{align}\label{eq:conditional-distribution}
    q_{\wh{M}M}(\wh{m}\vert m) :=  \tr\left[\Gamma^{\wh{m}}_{BF}\cE^m_{A\to A}(\proj{\Psi}_{AF})\right].
\end{align}
We call $\cC\equiv(\{\cE^m\}_m,\cD)\in(\fE,\fD)$
a \textit{dense coding code} for the resourceful quantum state $\Psi_{AF}$
under the available encoder-decoder pair $(\fE,\fD)$
with cardinality $\vert\cC\vert\equiv\vert\cM\vert$.
We quantify the performance of $\cC$ by computing the \textit{decoding error}:
\begin{align}\Label{eq:average probability of error}
    e(\cC) := 1 - \frac{1}{\vert\cM\vert}\sum_m q_{\wh{M}M}(m\vert m)
\end{align}
and use $s(\cC):=1-e(\cC)$ to denote the success probability of decoding. In general, smaller decoding error implies
better code. However, to achieve small $e(\cC)$ one has to encode with small size $\vert\cC\vert$. This motivates us
to define the \emph{dense coding rate} that quantitatively measures the communication capacity of the code:
\begin{align}
    r\left(\cC\right) := \log\vert\cC\vert.
\end{align}

\begin{figure}[!htbp]
\centering
  \includegraphics[width=0.8\textwidth]{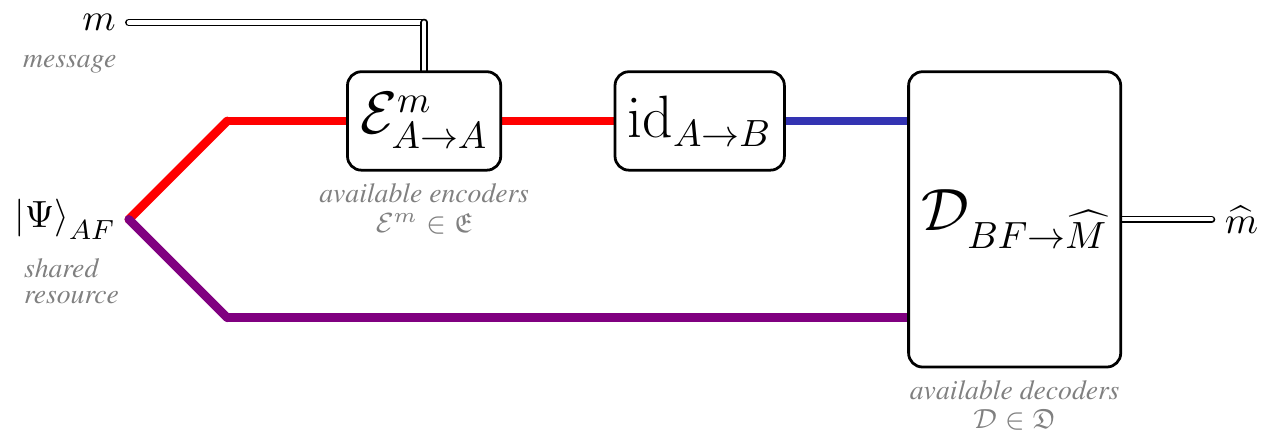}
  \caption{A dense coding protocol for the shared resourceful quantum state $\Psi_{AF}$
        under the available encoder-decoder pair $(\fE,\fD)$.
        In this protocol, Alice possesses the state in red line,
        Bob possesses the state in blue line, and
        Fred possesses the state in purple line.}
  \Label{fig:dense-coding}
\end{figure}

Fix $\varepsilon\in[0,1)$. The one-shot $\varepsilon$-dense coding capacity of $\Psi_{AF}$ under
available encoder-decoder pair $(\fE,\fD)$ is defined to be the
the maximum bits of messages that can be transmitted such that the decoding error is
upper bounded by the error threshold $\varepsilon$.

\begin{definition}[One-shot $\varepsilon$-dense coding capacity]
\label{def:one-shot dense coding capacity}
Let $\ket{\Psi}_{AF}$ be a bipartite pure quantum state and $\varepsilon\in[0,1)$. The one-shot
$\varepsilon$-dense coding capacity of $\Psi_{AF}$ under
available encoder-decoder pair $(\fE,\fD)$ is defined as:
\begin{align}\label{eq:one-shot dense coding capacity}
    C_{\fE,\fD}^\varepsilon\left(\Psi_{AF}\right)
:= \sup_{\cC \in (\fE,\fD)}\left\{r\left(\cC\right) \sbar e(\cC) \leq \varepsilon \right\}.
\end{align}
\end{definition}

The dense coding capacity of $\Psi_{AF}$ under available encoder-decoder pair $(\fE,\fD)$ is then
defined to be the one-shot $\varepsilon$-dense coding capacity of $\Psi_{AF}^{\ox n}$
by taking the limits $n\to\infty$ and $\varepsilon\to0$. This capacity
quantifies the ultimate number of bits that can be reliably transmitted per
copy of $\Psi_{AF}$ in the asymptotic regime, under available encoder-decoder pair $(\fE,\fD)$.

\begin{definition}[Dense coding capacity]
\label{def:dense coding capacity}
Let $\ket{\Psi}_{AF}$ be a bipartite pure quantum state. The dense coding capacity of $\Psi_{AF}$
under available encoder-decoder pair $(\fE,\fD)$ is defined as:
\begin{align}\Label{eq:dense coding capacity}
    C_{\fE,\fD}\left(\Psi_{AF}\right)
:=  {\adjustlimits\inf_{\varepsilon>0}\limsup_{n\to\infty}}
    \frac{1}{n}C_{\fE,\fD}^\varepsilon\left(\Psi_{AF}^{\ox n}\right).
\end{align}
\end{definition}

Analogously, the strong converse dense coding capacity of $\Psi_{AF}$ is defined to be the one-shot
$\varepsilon$-dense coding capacity of $\Psi_{AF}^{\ox n}$ by taking the limit $n\to\infty$
and satisfying the constraint that $\varepsilon<1$.
This capacity quantifies to what extend we can sacrifice the decoding error to achieve
larger dense coding rate in the asymptotic regime, under available encoder-decoder pair $(\fE,\fD)$.

\begin{definition}[Strong converse dense coding capacity]
\label{def:strong converse dense coding capacity}
Let $\ket{\Psi}_{AF}$ be a bipartite pure quantum state.
The strong converse dense coding capacity of $\Psi_{AF}$
under available encoder-decoder pair $(\fE,\fD)$ is defined as:
\begin{align}\Label{eq:strong converse dense coding capacity}
    C_{\fE,\fD}^\dagger\left(\Psi_{AF}\right)
:=  {\adjustlimits\sup_{\varepsilon<1}\limsup_{n\to\infty}}
    \frac{1}{n}C_{\fE,\fD}^\varepsilon\left(\Psi_{AF}^{\ox n}\right).
\end{align}
\end{definition}

In the following sections we introduce various available classes of encoders $\fE$
and decoders $\fD$ within the resource theory of asymmetry framework.

\subsection{Quantum and super-quantum encoders}

Practically, it is not easy to implement arbitrary quantum operations for an encoder.
Hence, it is natural to restrict Alice's encoding operation to a certain class of operations.
For example, when a Hamiltonian $H$ is fixed, the unitary operation $e^{it H}$ can be easily implemented.
Noticing that the set $\{e^{it H}\}_t$ forms a group representation,
this example can be generalized as follows.
Given a group $G$ and its unitary (projective) representation $U_g$,
we assume that the following available set of Alice's encoding operations
\begin{align}\Label{eq:free-operations}
  \EncG := \left\{ \cU_g \sbar g \in G\right\},
\end{align}
where $\cU_g$ is defined in~\eqref{eq:symmetric}.
We note that $\EncG$ remains as our first and smallest set of encoders.

However, in general, the dephasing operation can be easily experimentally implemented
yet it is not included in $\EncG$. This motivates us to enlarge $\EncG$ to encapsulate
physically implementable operations.
Following the argument outlined in~\cite[Section I]{korzekwa2019encoding},
we enlarge $\EncG$ by proposing the following available set of encoding operations
that \emph{commute} with $\cG$:
\begin{align}\Label{eq:free-operations2}
  \EncCP := \left\{\cE\in\channel{A\to A}_{\rm cp}\sbar \cE\circ\cG = \cG = \cG\circ\cE\right\},
\end{align}
where $\channel{A\to A}_{\rm cp}$ is the set of trace preserving and completely positive (TPCP) maps
from $A$ to $A$. This class is larger than $\EncG$.
In fact, when the Hamiltonian is given as $H$
and the set $\EncG$ is given as the set $\{e^{it H}\}_t$, the dephasing operation is contained in $\EncCP$,
matching our requirement.

To investigate the power of \emph{super-quantum} encoders, we introduce the following class of encoding operations:
\begin{align}\Label{eq:free-operations3}
  \EncP := \left\{\cE\in\channel{A\to A}_{\rm p}\sbar \cE\circ\cG = \cG = \cG\circ\cE\right\},
\end{align}
where $\channel{A\to A}_{\rm p}$ is the set of trace preserving and positive maps from $A$ to $A$.
From the perspective of resource theory of asymmetry, each quantum operation $\cE\in\EncP$
can encode information (both classical and quantum) into some degrees of
freedom of resourceful states that can be completely destroyed by $\cG$.
The set of positive maps from $2$-dimensional system to itself
is generated by the TPCP maps and the transpose operation~\cite{horodecki1996separability}.
However, \cite{horodecki1997separability,skowronek2016there} showed that
the set of positive maps from $3$-dimensional system to itself
requires infinitely many generators.
This fact indicates the possibility of enhancing the dense coding capacity by
using the set of encoders $\EncP$ over $\EncCP$.

Since $\channel{A\to A}_{\rm p}$ is strictly larger than
$\channel{A\to A}_{\rm cp}$, we can consider an intermediate set
$\channel{A\to A}_{\rm ppt}$ defined as
\begin{align}
\channel{A\to A}_{\rm ppt}:=
\left\{ \cE \in \channel{A\to A}_{\rm p} \sbar (\id_F\ox\cE)(\proj{\Phi}) \hbox{ is a PPT state} \right\},
\end{align}
where $\Phi$ is a maximally entangled state on the bipartite system $\cH_{AF}$.
As another virtual setting, we may also consider the following set of encoders:
\begin{align}\Label{eq:free-operations3}
  \EncPPT := \left\{\cE\in\channel{A\to A}_{\rm ppt}\sbar \cE\circ\cG = \cG = \cG\circ\cE\right\}.
\end{align}
Studying the above two classes of encoders, we can clarify whether super-quantum encoders
can enhance classical information transmission with pre-shared resourceful quantum state.

We can show the following inclusion hierarchy for the four classes of encoders defined above:
\begin{align}
  \EncG \subset \EncCP \subset \EncPPT \subset \EncP.
\end{align}

\subsection{Decoders under locality conditions}

Many studies investigated dense coding protocols under the assumptions that
entanglement are preshared and arbitrary encoding and decoding operations are allowed.
However, even when the sender and the receiver share an entangled state,
it is not so easy to implement a general joint measurement across two
quantum systems---the message receiver $\cH_B$ and the entanglement receiver $\cH_F$.
Hence, it is natural to impose \emph{locality conditions} for the decoders.

As a typical case, we can consider the \emph{one-way LOCC decoders} $\DecOne$ where the
classical communication flows from the entanglement receiver $\cH_F$ to the message receiver $\cH_B$.
In this case, the entanglement receiver Fred first measures the shared state at hand
and then shares the information to the message receiver Bob via classical communication.
Conditioned on the information, Bob decodes the message using local decoders.
Aiming to improve the dense coding capacity, we also introduce the \emph{two-way LOCC decoders} $\DecLOCC$
where the two-way LOCC operations can be realized
by the combinations of local operations and classical communications between the two systems.
We remark that (one-way) LOCCs are the most natural set of quantum operations in the
distributed quantum information processing.

Motivated by the resource theory of quantum entanglement~\cite{horodecki2009quantum},
we may further enlarge the set of available decoder POVMs w.r.t. the bipartite system
$B{:}F$ to improve the communication rate by considering \emph{separable decoders}
(also known as separable measurements) $\DecSEP$
and \emph{PPT decoders} (also known as PPT measurements) $\DecPPT$.
Intuitively, a joint measurement is a separable (PPT) measurement if all of its
POVM elements can be implemented by separable operations (PPT operations).
Although the separable decoders $\DecSEP$ and
the PPT decoders $\DecSEP$ are theoretical objects,
they are useful in proving the converse part in coding theorems.

At last we consider two special cases that cover the commonly studied dense coding tasks.
First, we investigate the \emph{local decoders} $\DecL$ in which Bob decodes the message encoded by
Alice without the help from Fred (the entanglement receiver),
which we shall call the non-assisted decoding.
Second, we investigate the \emph{global decoders} $\DecG$ in which Bob and Fred work together
to decode the message encoded by Alice. We do not impose any locality condition on the
joint measurements they can carry out.
Note that global decoders are used in the seminal dense coding protocol
originally proposed by Bennett and Wiesner~\cite{bennett1992communication},

We conclude the following inclusion hierarchy for the six classes of decoders defined above:
\begin{align}
  \DecL \subset \DecOne \subset \DecLOCC \subset \DecSEP \subset \DecPPT \subset \DecG.
\end{align}

\subsection{Landscape of dense coding capacities}

In the above two sections we have proposed four classes of available
encoders---$\EncG, \EncCP, \EncPPT, \EncP$---and six classes of available
decoders---$\DecL, \DecOne, \DecLOCC, \DecSEP, \DecPPT, \DecG$.
However, it is not the case that arbitrary encoder-decoder pair chosen from the available sets
can form an valid code for the resourceful quantum state $\Psi_{AF}$.
More precisely, consider the encoder-decoder pair $(\fE,\fD)$ where
 $\fE\in\{\EncG, \EncCP, \EncPPT, \EncP\}$
and $\fD\in\{\DecL,\DecOne, \DecLOCC, \DecSEP, \DecPPT,\DecG\}$.
The encoding operations $\{\cE^m\}_m$ chosen from $\fE$ by Alice
and the decoder POVM $\{\Gamma^{m}\}_m$ chosen from $\fD$ by Bob and Fred
yield the conditional values $q_{\wh{M}M}$ defined in Eq.~\eqref{eq:conditional-distribution}.
To guarantee that $q_{\wh{M}M}$ is a conditional distribution
(and the corresponding code $\cC=(\{\cE^m\}_m,\{\Gamma^{m}\}_m)$ is an valid dense coding code),
it must hold that
\begin{align}
    q_{\wh{M}M}(\wh{m}\vert m) = \tr\left[\Gamma^{\wh{m}}_{BF}\cE^m_{A\to A}(\proj{\Psi}_{AF})\right] \geq 0
\end{align}
for arbitrary $m,\wh{m}\in\cM$. This physical constraint rules out
the possible combinations $(\EncPPT,\DecG), (\EncP,\DecPPT), (\EncP,\DecG)$,
since these encoder-decoder pairs may lead to negative values.
Conversely, all other possible pairs $(\fE,\fD)$
($21$ pairs in total) are valid encoder-decoder pairs for the dense coding protocol.
For reference, we outline the landscape of investigated dense coding capacities
in Table~\ref{table:landscape}.

\begin{remark}
On the extreme case where the measurement outcomes on $\cH_F$ is completely \emph{ignored},
i.e., only local decoders $\DecL$ are available, we recover the
communication capacities of $\Psi_A$ previously investigated in~\cite{korzekwa2019encoding}.
\end{remark}

\begin{table}[!hptb]
\centering
\renewcommand{\multirowsetup}{\centering}
\renewcommand*{\arraystretch}{2}
\setlength\tabcolsep{6.4pt} 
\begin{tabular}{ccccccccc} 
\toprule[1.5pt]
& & & & \multicolumn{5}{c}{\textbf{Encoder $\fE$}} \\
\cmidrule{5-9}
& & & & \multicolumn{2}{c}{Quantum Encoder} &  & \multicolumn{2}{c}{Super-Quantum Encoder} \\
\cmidrule{5-6} \cmidrule{8-9}
& & & & $\EncG$ & $\EncCP$ & & $\EncPPT$ & $\EncP$ \\\midrule[0.5pt]
\multirow{6}{*}{\textbf{Decoder $\fD$}}
& Local & $\DecL$ &
              & {\color{blue}$C_{\EncG, \DecL}(\Psi_{AF})$}
              & {\color{blue}$C_{\EncCP, \DecL}(\Psi_{AF})$}
              &
              & {\color{blue}$C_{\EncPPT, \DecL}(\Psi_{AF})$}
              & {\color{blue}$C_{\EncP, \DecL}(\Psi_{AF})$} \\\cmidrule{2-9}
& One-way LOCC & $\DecOne$ &
              & {\color{red}$C_{\EncG, \DecOne}(\Psi_{AF})$}
              & {\color{red}$C_{\EncCP, \DecOne}(\Psi_{AF})$}
              &
              & {\color{red}$C_{\EncPPT, \DecOne}(\Psi_{AF})$}
              & {\color{red}$C_{\EncP, \DecOne}(\Psi_{AF})$} \\\cmidrule{2-9}
& LOCC & $\DecLOCC$ &
              & {\color{red}$C_{\EncG, \DecLOCC}(\Psi_{AF})$}
              & {\color{red}$C_{\EncCP, \DecLOCC}(\Psi_{AF})$}
              &
              & {\color{red}$C_{\EncPPT, \DecLOCC}(\Psi_{AF})$}
              & {\color{red}$C_{\EncP, \DecLOCC}(\Psi_{AF})$} \\\cmidrule{2-9}
& Separable & $\DecSEP$ &
              & {\color{red}$C_{\EncG, \DecSEP}(\Psi_{AF})$}
              & {\color{red}$C_{\EncCP, \DecSEP}(\Psi_{AF})$}
              &
              & {\color{red}$C_{\EncPPT, \DecSEP}(\Psi_{AF})$}
              & {\color{red}$C_{\EncP, \DecSEP}(\Psi_{AF})$} \\\cmidrule{2-9}
& PPT & $\DecPPT$ &
              & {\color{red}$C_{\EncG, \DecPPT}(\Psi_{AF})$}
              & {\color{red}$C_{\EncCP, \DecPPT}(\Psi_{AF})$}
              &
              & {\color{red}$C_{\EncPPT, \DecPPT}(\Psi_{AF})$}
              & {\large\ding{55}} \\\cmidrule{2-9}
& Global  & $\DecG$ &
              & {\color{violet}$C_{\EncG, \DecG}(\Psi_{AF})$}
              & {\color{violet}$C_{\EncCP, \DecG}(\Psi_{AF})$}
              &
              & {\large\ding{55}}
              & {\large\ding{55}} \\
\bottomrule[1.5pt]
\end{tabular}
\caption{Landscape of the dense coding capacities investigated in this paper.
    We are able to show that the capacities with the same color are actually equal
    and derive single-letter capacity formulas for all these capacities under Assumption~\ref{assp:multiplicity-free}.}
\label{table:landscape}
\end{table}

We summarize in the following proposition the inclusion relations
of the various dense coding capacities defined above.

\begin{proposition}\Label{prop:enhanced-relation2}
Let $\ket{\Psi}_{AF}$ be a bipartite pure quantum state. It holds that
 \begin{align}
C_{\fE,\fD}(\Psi_{AF})
\le C^\dagger_{\fE,\fD}(\Psi_{AF})
\end{align}
for $\fE\in\{\EncG, \EncCP, \EncPPT, \EncP\}$
and $\fD\in\{\DecL,\DecOne, \DecLOCC, \DecSEP, \DecPPT,\DecG\}$,
except for $(\EncPPT,\DecG), (\EncP,\DecPPT), (\EncP,\DecG)$.
What's more, the following relation hierarchy holds:

\begin{center}
\begin{tikzpicture}[scale=1]

\node at (0,0) {$C_{\EncG,\DecL}$};%
\node at (3,0) {$C_{\EncCP,\DecL}$};%
\node at (6,0) {$C_{\EncPPT,\DecL}$};%
\node at (9,0) {$C_{\EncP,\DecL}$};%

\node at (0,-2) {$C_{\EncG,\DecOne}$};%
\node at (3,-2) {$C_{\EncCP,\DecOne}$};%
\node at (6,-2) {$C_{\EncPPT,\DecOne}$};%
\node at (9,-2) {$C_{\EncP,\DecOne}$};%

\node at (0,-4) {$C_{\EncG,\DecLOCC}$};%
\node at (3,-4) {$C_{\EncCP,\DecLOCC}$};%
\node at (6,-4) {$C_{\EncPPT,\DecLOCC}$};%
\node at (9,-4) {$C_{\EncP,\DecLOCC}$};%

\node at (0,-6) {$C_{\EncG,\DecSEP}$};%
\node at (3,-6) {$C_{\EncCP,\DecSEP}$};%
\node at (6,-6) {$C_{\EncPPT,\DecSEP}$};%
\node at (9,-6) {$C_{\EncP,\DecSEP}$};%

\node at (0,-8) {$C_{\EncG,\DecPPT}$};%
\node at (3,-8) {$C_{\EncCP,\DecPPT}$};%
\node at (6,-8) {$C_{\EncPPT,\DecPPT}$};%

\node at (0,-10) {$C_{\EncG,\DecG}$};%
\node at (3,-10) {$C_{\EncCP,\DecG}$};%

\draw[thick,->] (1,0) -- (2,0);
\draw[thick,->] (4,0) -- (5,0);
\draw[thick,->] (7,0) -- (8,0);
\draw[thick,->] (1,-2) -- (2,-2);
\draw[thick,->] (4,-2) -- (5,-2);
\draw[thick,->] (7,-2) -- (8,-2);
\draw[thick,->] (1,-4) -- (2,-4);
\draw[thick,->] (4,-4) -- (5,-4);
\draw[thick,->] (7,-4) -- (8,-4);
\draw[thick,->] (1,-6) -- (2,-6);
\draw[thick,->] (4,-6) -- (5,-6);
\draw[thick,->] (7,-6) -- (8,-6);
\draw[thick,->] (1,-8) -- (2,-8);
\draw[thick,->] (4,-8) -- (5,-8);
\draw[thick,->] (1,-10) -- (2,-10);

\draw[thick,->] (0,-0.5) -- (0,-1.5);
\draw[thick,->] (3,-0.5) -- (3,-1.5);
\draw[thick,->] (6,-0.5) -- (6,-1.5);
\draw[thick,->] (9,-0.5) -- (9,-1.5);

\draw[thick,->] (0,-2.5) -- (0,-3.5);
\draw[thick,->] (3,-2.5) -- (3,-3.5);
\draw[thick,->] (6,-2.5) -- (6,-3.5);
\draw[thick,->] (9,-2.5) -- (9,-3.5);

\draw[thick,->] (0,-4.5) -- (0,-5.5);
\draw[thick,->] (3,-4.5) -- (3,-5.5);
\draw[thick,->] (6,-4.5) -- (6,-5.5);
\draw[thick,->] (9,-4.5) -- (9,-5.5);

\draw[thick,->] (0,-6.5) -- (0,-7.5);
\draw[thick,->] (3,-6.5) -- (3,-7.5);
\draw[thick,->] (6,-6.5) -- (6,-7.5);

\draw[thick,->] (0,-8.5) -- (0,-9.5);
\draw[thick,->] (3,-8.5) -- (3,-9.5);
\end{tikzpicture}
\end{center}
where $x\rightarrow y$ means that $x\leq y$.
Also, we have the same relation hierarchy
for the strong converse capacities $ C^\dagger_{\fE,\fD}(\Psi_{AF})$.
\end{proposition}

\subsection{Enhanced version with one-way LOCC}\label{sec:one-way}

In the above dense coding framework, if we fix the available decoders to $\DecOne$,
i.e., the one-way LOCC decoders, this specific setting has an equivalent description called the
\textit{environment-assisted classical communication via quantum resources},
originally motivated by the intensively studied environment assistance
framework~\cite{divincenzo1998entanglement,gregoratti2003quantum,hayden2004correcting,smolin2005entanglement,winter2005environment,buscemi2005inverting,buscemi2007channel,dutil2010assisted,buscemi2013general,karumanchi2016classical,karumanchi2016quantum,chitambar2016assisted,regula2018nonasymptotic,vijayan2018one,lami2020assisted}.
The detailed dense coding procedure using one-way LOCC decoders is illustrated in Figure~\ref{fig:ea-dense-coding}.

\begin{figure}[!htbp]
\centering
\begin{minipage}[t]{.48\linewidth}
  \centering
  \includegraphics[width=1.0\textwidth]{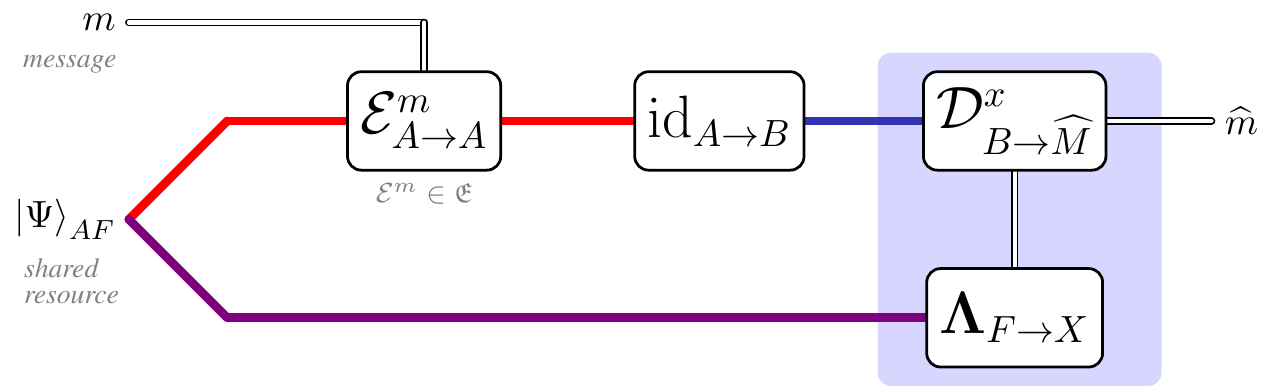}
  \caption{A dense coding protocol for the shared resourceful quantum state $\Psi_{AF}$
            under the one-way LOCC decoders (shaded area).
            In this protocol, Alice possesses the state in red line,
            Bob possesses the state in blue line, and
            Fred possesses the state in purple line.}
  \Label{fig:ea-dense-coding}
\end{minipage}%
\hspace{0.1in}
\begin{minipage}[t]{.49\linewidth}
  \centering
  \includegraphics[width=1.0\textwidth]{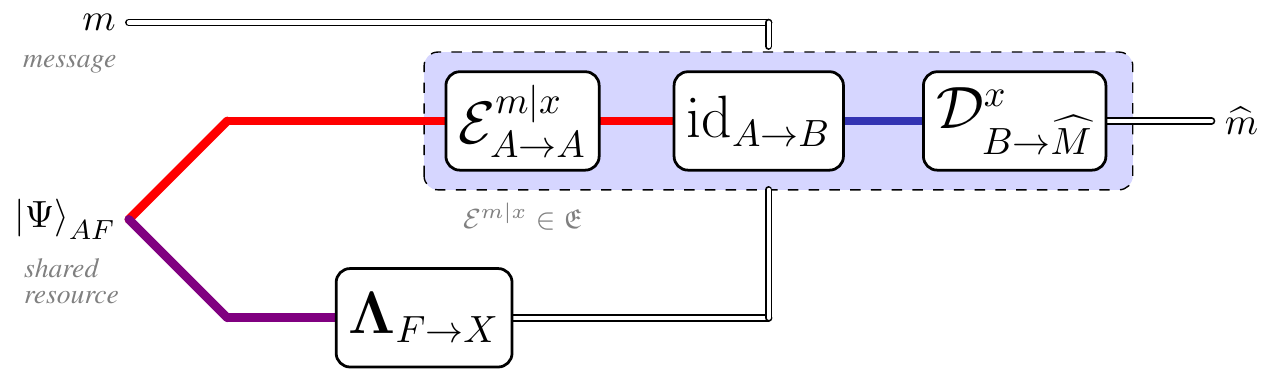}
  \caption{An enhanced dense coding protocol for the shared resourceful quantum state $\Psi_{AF}$
            under the one-way LOCC decoders.
            In this protocol, Alice possesses the state in red line,
        Bob possesses the state in blue line, and
        Fred possesses the state in purple line.}
  \Label{fig:ea-dense-coding-enhanced}
\end{minipage}
\end{figure}

Inspired by Figure~\ref{fig:ea-dense-coding},
we propose here a hypothetical and \emph{enhanced} dense coding framework with one-way LOCC decoders
in which both Alice and Bob have access to Fred's outcome, as illustrated in Figure~\ref{fig:ea-dense-coding-enhanced}.
This hypothetical setting yields upper bounds on the standard dense coding with one-way LOCC decoders.
Obviously, the dense coding power of $\Psi_{AF}$ is enhanced compared to the
setting depicted in Figure~\ref{fig:ea-dense-coding} since Alice
possesses additional information (from Fred).
Following the Definitions~\ref{def:one-shot dense coding capacity},~\ref{def:dense coding capacity},
and~\ref{def:strong converse dense coding capacity}, we can define analogously
corresponding enhanced dense coding capacities introduced in Figure~\ref{fig:ea-dense-coding-enhanced} as
\begin{align}\label{eq:enchance dense coding}
    \wt{C}_{\fE,\DecOne}^\varepsilon(\Psi_{AF}),\quad
    \wt{C}_{\fE,\DecOne}(\Psi_{AF}),\quad
    \wt{C}_{\fE,\DecOne}^\dagger(\Psi_{AF}),
\end{align}
respectively, where $\fE\in\{\EncG,\EncCP,\EncPPT,\EncP\}$.
Throughout this paper, we use the letter $\wt{C}$ with tilde to represent the enhanced dense coding capacity.
We conclude the following weak and strong converse bounds on the (enhanced) dense coding capacities.
See Appendix~\ref{appx:thm:enhanced-relation} for the proof.

\begin{theorem}\Label{thm:enhanced-relation}
Let $\ket{\Psi}_{AF}$ be a bipartite pure quantum state and $\varepsilon\in[0,1)$. It holds that
\begin{subequations}
\begin{align}
  C_{\fE,\DecOne}^\varepsilon(\Psi_{AF})
&\leq \wt{C}_{\fE,\DecOne}^\varepsilon(\Psi_{AF}),\label{eq:enhanced-relation-1} \\
  C_{\fE,\DecOne}(\Psi_{AF}) &\leq \wt{C}_{\fE,\DecOne}(\Psi_{AF})
\leq A_G^\infty(\Psi_A),\label{eq:enhanced-relation-2}\\
  C^\dagger_{\fE,\DecOne}(\Psi_{AF}) &\leq \wt{C}^\dagger_{\fE,\DecOne}(\Psi_{AF})
\leq \Shannon\left(\cG(\Psi_A)\right)\label{eq:enhanced-relation-3},
\end{align}
where $\fE\in\{\EncG,\EncCP,\EncPPT,\EncP\}$.
\end{subequations}
\end{theorem}

\subsection{Main results}

\paragraph*{Dense coding capacities under locality conditions.}
Our main result concerns the dense coding capacities under various locality conditions---$\DecOne$,
$\DecLOCC$, $\DecSEP$, $\DecPPT$. In a word, we show that all these capacities are equal
and derive a single-letter capacity formula.
Before stating the result, we outline some notations first.
We assume that the (projective) unitary representation $U$ on $\cH_A$ is
multiplicity-free (cf. Assumption~\ref{assp:multiplicity-free}).
The Hilbert space $\cH_A$ is decomposed as $\oplus_{k \in \cK} \cH_k$.
Hence, any pure state $\Psi_{AF}$ on the bipartite system $\cH_A\ox\cH_F$
can be written as
\begin{align}
\Psi_{AF}= \sum_{k \in \cK} \sqrt{P_K(k)}\Psi_{AF,k},
\end{align}
where $\Psi_{AF,k}$ is a pure state on the bipartite system $\cH_k \ox\cH_F$.
The average state on the bipartite system $\cH_A\ox\cH_F$ is given as
\begin{align}\label{eq:xi-AF}
\xi_{AF}:=(\cG_A\ox\id_F)(\Psi_{AF})
= \sum_{k \in \cK}P_K(k) \pi_k\ox \rho_{F,k},
\end{align}
where $\rho_{F,k}:= \tr_A\Psi_{AF,k}$ and $\pi_k$ is the maximally mixed state on $\cH_k$.
Notice that
\begin{align}\label{eq:QQbwQRNzEZSI}
\Shannon(\xi_{AF})=\Shannon(K)_\xi+ \Shannon(A|K)_\xi+\Shannon(F|K)_\xi
=\Shannon(A)_\xi+\Shannon(F|K)_\xi,
\end{align}

Our main result is summarized as follows and the proof can be found in Appendix~\ref{S4B}.

\begin{theorem}[Dense coding capacity under locality conditions]
\Label{thm:asymptotic characterization}
Let $\ket{\Psi}_{AF}$ be a bipartite pure quantum state.
It holds under Assumption~\ref{assp:multiplicity-free} (the multiplicity-free condition) that
\begin{align}\Label{eq:asymptotic characterization}
C_{\fE,\fD}(\Psi_{AF})
=C^\dagger_{\fE,\fD}(\Psi_{AF})
= A_G^\infty(\Psi_A) = \Shannon\left(\cG(\Psi_A)\right)
= \Shannon\left(A\right)_{\xi},
\end{align}
for $\fE\in\{\EncG, \EncCP, \EncPPT, \EncP\}$
and $\fD\in(\DecOne, \DecLOCC, \DecSEP, \DecPPT)$,
except for $(\EncP,\DecPPT)$, where $\xi_A$ is defined in Eq.~\eqref{eq:xi-AF}.
\end{theorem}

\begin{remark}
Theorem~\ref{thm:asymptotic characterization} reveals the fact that
even when we enlarge the available encoder-decoder pair up to
$(\EncP, \DecSEP)$ or $(\EncPPT, \DecPPT)$,
we cannot improve the dense coding capacity compared to minimal encoder-decoder pair $(\EncG,\DecOne)$,
where the available encoders are the unitary encoding operations
and the available decoders are the one-way LOCC measurements.
\end{remark}

\begin{remark}
Theorem~\ref{thm:asymptotic characterization} shows that the dense coding capacities
under locality conditions all satisfy the desirable \emph{strong converse property}.
That is, for arbitrary dense coding code $\cC\in(\fE,\fD)$, where
$\fE\in\{\EncG, \EncCP, \EncPPT, \EncP\}$ and $\fD\in(\DecOne, \DecLOCC, \DecSEP, \DecPPT)$
except for $(\EncP,\DecPPT)$,
the decoding error necessarily converges to one in the asymptotic limit
whenever the coding rate exceeds the optimal rate $\Shannon\left(\cG(\Psi_A)\right)$.
We thus conclude that $\Shannon\left(\cG(\Psi_A)\right)$ is a very sharp
dividing line between the coding rates those are achievable and those are not,
ruling out the possibility of error-rate tradeoff in this dense coding task.
\end{remark}

\begin{remark}
Theorem~\ref{thm:asymptotic characterization} establishes an interesting equivalence
among three different quantities at the first glance: the operationally defined dense coding
capacity $C_{\fE,\fD}(\Psi_{AF})$~\eqref{eq:dense coding capacity},
the mathematically defined regularized asymmetry of assistance
$A_G^\infty\left(\Psi_A\right)$~\eqref{eq:egularized-asymmetry-of-assistance},
and the quantum entropy of the twirled quantum state $\Shannon\left(\cG(\Psi_A)\right)$.
See the triangle associations depicted in Figure~\ref{fig:relation-among-quantities}.
In this way, we provide the asymmetry measure $A_G^\infty$ with an operational meaning
in terms of the dense coding tasks.

\begin{figure}[htbp]
  \centering
  \includegraphics[width=0.5\textwidth]{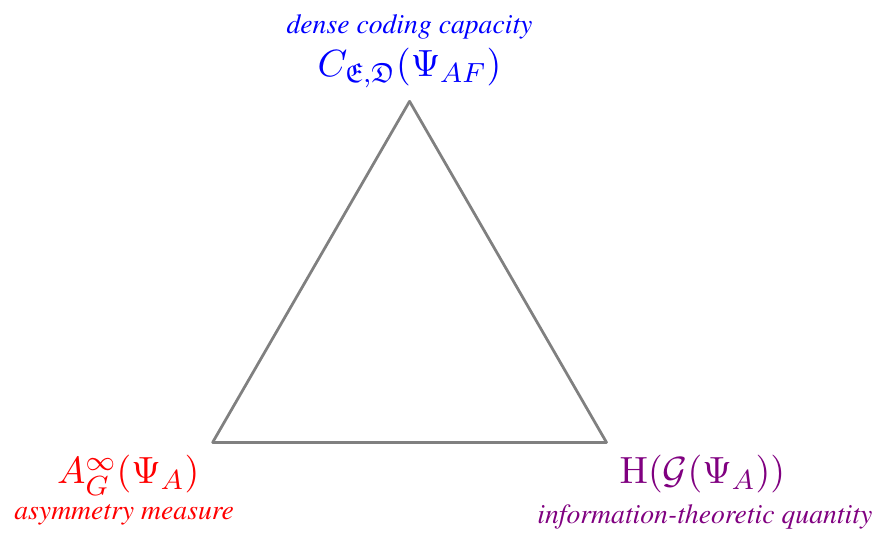}
  \caption{Equivalence among different quantities of $\Psi_{AF}$:
          the operationally defined dense coding capacity $C_{\fE,\fD}(\Psi_{AF})$,
          the mathematically defined regularized asymmetry of assistance $A_G^\infty(\Psi_A)$,
          and the information-theoretic quantity $\Shannon(\cG(\Psi_A))$.}
  \label{fig:relation-among-quantities}
\end{figure}
\end{remark}

As a direct corollary of Theorems~\ref{thm:enhanced-relation} and~\ref{thm:asymptotic characterization},
we conclude that even if Alice has access to the measurement outcomes sent by Fred
(cf. Figure~\ref{fig:ea-dense-coding-enhanced} for the enhanced dense coding framework),
the dense coding capability cannot be improved when compared to the standard dense coding framework
where only Bob can access to the measurement outcomes sent by Fred, when one-way LOCC decoders are available.
This, in some sense, indicates that Alice can choose the encoding operations completely
independent of the encoded state and yields an ``universal encoding'' strategy.

\begin{corollary}\Label{thm:asymptotic characterization2}
Let $\ket{\Psi}_{AF}$ be a bipartite pure quantum state. It holds that
\begin{align}
\wt{C}_{\fE,\DecOne}(\Psi_{AF})
= \wt{C}^\dagger_{\fE,\DecOne}(\Psi_{AF})
= A_G^\infty(\Psi_A) = \Shannon\left(\cG(\Psi_A)\right)
= \Shannon\left(A\right)_{\xi},
\end{align}
where $\fE\in\{\EncG,\EncCP,\EncPPT,\EncP\}$.
\end{corollary}

\paragraph*{Dense coding capacities with local decoders.}
When only the local decoders $\DecL$ are available,
we can derive the following coding theorem.
Notice that our results on the local decoders recover~\cite[Theorem 3]{korzekwa2019encoding}
as a special case, where they have considered the case where $\fE=\EncCP$.
See Appendix~\ref{S6} for the proof.

\begin{theorem}[Dense coding capacity with local decoders]\Label{thm:asymptotic empty}
Let $\ket{\Psi}_{AF}$ be a bipartite pure quantum state and let $\Psi_A:=\tr_F\Psi_{AF}$.
It holds under Assumption~\ref{assp:multiplicity-free} (the multiplicity-free condition) that
\begin{align}\label{eq:asymptotic empty}
  C_{\fE,\DecL}(\Psi_{AF})
= C^\dagger_{\fE,\DecL}(\Psi_{AF})
= \Shannon\left(\cG(\Psi_{A})\right)-\Shannon\left(\Psi_{A}\right)
= \Rel\left(\Psi_{A}\| \cG(\Psi_{A})\right),
\end{align}
for $\fE\in\{\EncG, \EncCP, \EncPPT, \EncP\}$.
\end{theorem}

\paragraph*{Dense coding capacities with global decoders.}
When the global decoders $\DecG$ are available,
we actually identify an variant of the well-known dense coding
task~\cite{bennett1992communication,hiroshima2001optimal,bowen2001classical,horodecki2001classical,winter2002scalable,bruss2004distributed,beran2008nonoptimality,horodecki2012quantum,datta2015second,laurenza2019dense,wakakuwa2020superdense}
in which the available encoders in the system $A$ are constrained by the twirling operation $\cG_A$.
Notice that the following result has previously been discovered in~\cite[Theorem 3]{korzekwa2019encoding}.
\begin{proposition}[Dense coding capacity with global decoders~\cite{korzekwa2019encoding}]\Label{prop:Global decodings}
Let $\ket{\Psi}_{AF}$ be a bipartite pure quantum state. It holds that
\begin{align}
  C_{\EncG,\DecG}(\Psi_{AF})
= C_{\EncCP,\DecG}(\Psi_{AF})
= \Shannon\left((\cG_A\ox\id_F)(\Psi_{AF})\right).
\end{align}
\end{proposition}

\begin{remark}
By Eq.~\eqref{eq:QQbwQRNzEZSI}, Theorem \ref{thm:asymptotic empty} can be rewritten as
\begin{align}\Label{eq:asymptotic empty2}
  C_{\fE,\DecL}(\Psi_{AF})
= C^\dagger_{\fE,\DecL}(\Psi_{AF})
= \Shannon\left(\cG(\Psi_{A})\right)-\Shannon\left(\Psi_{A}\right)
= \Shannon(A)_\xi-\Shannon(F)_\xi.
\end{align}
And also, Proposition \ref{prop:Global decodings} can be rewritten as
\begin{align}\label{eq:global decoding capacity}
  C_{\EncG,\DecG}(\Psi_{AF})
= C_{\EncCP,\DecG}(\Psi_{AF})
= \Shannon\left((\cG_A\ox\id_F)(\Psi_{AF})\right)
= \Shannon(A)_\xi + \Shannon(F|K)_\xi.
\end{align}
Writing $\Shannon\left(A\right)_{\xi}$
as $\Shannon\left(A\right)_{\xi} = \Shannon\left(A\right)_{\xi} - \Shannon(F)_\xi + \Shannon(F)_\xi$
and comparing Eqs.~\eqref{eq:asymptotic characterization} and~\eqref{eq:asymptotic empty2},
we can see that the dense coding capacity under locality conditions can be interpreted as the sum of
the amount of asymmetry $\Shannon(A)_\xi-\Shannon(F)_\xi$ reserved in the quantum state $\Psi_A$
and the amount of assistance $\Shannon(F)_\xi$ from the non-local
decoders $\fD\in\{\DecOne, \DecLOCC, \DecSEP, \DecPPT\}$.
On the other hand, Eqs.~\eqref{eq:asymptotic characterization} and~\eqref{eq:global decoding capacity} together
imply that that $\Shannon(F|K)_\xi$ can be viewed as the merit of global decoders.
\end{remark}

\section{Examples}\Label{sec:examples}

In this section, we compute the dense coding capacities
$C_{\EncG,\DecL}(\Psi_{AF}) ,C_{\EncG,\DecOne}(\Psi_{AF}),C_{\EncG,\DecG}(\Psi_{AF})$
for specialized resource theories of asymmetry of practical interest.
Notice that by Theorems~\ref{thm:asymptotic characterization}
and~\ref{thm:asymptotic empty} and Proposition~\ref{prop:Global decodings},
it suffices to evaluate these three capacities.

\subsection{Dense coding power of purity}
First, we consider the case when the (projective) unitary representation $U$ is \textit{irreducible} on $\cH_A$.
For example, $G$ can be the group of unitary matrices on $\cH_A$.
Also, when $G$ is the discrete Weyl-Heisenberg group on $\cH_A$,
the corresponding $U$ forms a irreducible projective unitary representation.
In this case, the twirling operation $\cG$ becomes the completely depolarizing channel
such that $\cG(\rho_A) = \1_A/d_A$ for all $\rho_A\in\density{\cH_A}$,
where $d_A$ is the dimension of system $A$. Correspondingly, the induced resource theory
is known as the resource theory of purity~\cite{horodecki2003reversible,gour2015resource,streltsov2016maximal}.
For this resource theory, we have
\begin{subequations}
\begin{align}
    C_{\EncG,\DecL}(\Psi_{AF}) &= \log d_A - \Shannon(\Psi_A),\Label{eq:full1} \\
    C_{\EncG,\DecOne}(\Psi_{AF}) &= \log d_A,\Label{eq:full2} \\
    C_{\EncG,\DecG}(\Psi_{AF}) &= \log d_A + \Shannon(\Psi_A),\Label{eq:full3}
\end{align}
\end{subequations}
where the first equality was previously concluded in~\cite[Eq.~(13b)]{korzekwa2019encoding}, the second equality follows
from Theorem~\ref{thm:asymptotic characterization}, and the last equality follows from Proposition~\ref{prop:Global
decodings}. Comparing Eqs.~\eqref{eq:full1}-\eqref{eq:full3}, we obtain a strict communication power hierarchy
among different classes of decoders in the dense coding task, whenever $\Psi_A$
is \emph{mixed} and thus its quantum entropy is strictly positive:
\begin{align}
    C_{\EncG,\DecL}(\Psi_{AF}) < C_{\EncG,\DecOne}(\Psi_{AF}) < C_{\EncG,\DecG}(\Psi_{AF}).\Label{eq:full3R}
\end{align}

\subsection{Dense coding power of coherence}

When $G$ is a group of unitaries diagonal in a given basis $\{\ket{b}\}$ of system $A$ (i.e., it is a
subgroup of commuting unitaries), $\cG$ becomes the completely dephasing channel
$\Delta(\rho_A) = \sum_b\bra{b}\rho_A\ket{b}\proj{b}$ for all $\rho_A\in\density{\cH_A}$
and we recover the intensively studied resource theory of coherence~\cite{streltsov2017colloquium}.
In this case, the encoders do not change the diagonal elements of the quantum state in the given basis
$\cB$ but only affect the off-diagonal elements.
Investigating the dense coding task under this resource theory corresponds to
asking how much classical information can be encoded into the quantum coherence resource.
We have the following results
\begin{subequations}
\begin{align}
    C_{\EncG,\DecL}(\Psi_{AF}) &= \Shannon(\Delta(\Psi_A)) - \Shannon(\Psi_A),\Label{eq:coherence1} \\
    C_{\EncG,\DecOne}(\Psi_{AF}) &= \Shannon(\Delta(\Psi_A)),\Label{eq:coherence2} \\
    C_{\EncG,\DecG}(\Psi_{AF}) &= \Shannon(\Delta(\Psi_A)),\Label{eq:coherence3}
\end{align}
\end{subequations}
where the first equality was previously concluded in~\cite[Eq.~(19b)]{korzekwa2019encoding}, the second equality follows
from Theorem~\ref{thm:asymptotic characterization}, and the last equality follows from Proposition~\ref{prop:Global
decodings} and the fact that~\cite[Theorem 4]{chitambar2016assisted}
\begin{align}
    \Rel\left(\Psi_{AF}\rel\Delta_A(\Psi_{AF})\right)=\Shannon(\Delta(\Psi_A)),
\end{align}
whenever $\Psi_{AF}$ is pure.
The fact that $C_{\EncG,\DecOne}(\Psi_{AF})=C_{\EncG,\DecG}(\Psi_{AF})$ remarkably shows that
global decoding has no advantage over the one-way LOCC decoding
for the dense coding task within the resource theory of quantum coherence.
When the reduced density $\Psi_{A}$ is diagonal, $C_{\EncG,\DecL}(\Psi_{AF})$ evaluates to $0$,
indicating that incoherent quantum states have no communication power under our setting.

Indeed, the above discussion can be applied even when
the group $G$ is the one-dimensional group $\mathbb{R}$ in the following case.
Consider a diagonal Hermitian operator $H$ whose diagonal elements are different.
Then, we consider the unitary representation of $\mathbb{R}$
as $x \mapsto e^{i H}$.
Each one-dimensional space generated by diagonal element
is different irreducible component.
Hence, we can apply the above discussion.
When each diagonal element of $H$ is an integer,
the above can be considered as the unitary representation of the compact group
$[0,2\pi)$.

We remark that~\cite[Section IV]{shi2020practical} addressed the special case of
$C_{\EncG,\DecG}(\Psi_{AF})$ when $\Psi_{AF}$ is a two-mode squeezed vacuum (TMSV) state
\begin{align}
  \ket{\Psi}_{AF} = \sum_{n=0}^{\infty} \sqrt{N^n/(N+1)^{n+1}} \ket{n}_A\ket{n}_F.
\end{align}
For this special case, we have $\Shannon(\Delta(\Psi_A)) =\Shannon(\Psi_A)=(N+1)\log (N+1)-N\log N$.
Hence, if only local decoders $\DecL$ are available, Alice cannot transmit classical information to Bob
via the quantum state $\Psi_{AF}$.

\subsection{Dense coding power of Schur duality}

Assume now that $\cH_A=\cH^{\ox N}$, where $\cH$ is a $d$-dimensional Hilbert space.
That is, $A$ is $n$-partite system with equal local dimensions $d$.
The group $\unitary{\cH}$ has the unitary representation
$\{U^{\ox N}\}_{U\in\unitary{\cH}}$ on $\cH_A$.
Let $S(N)$ be the set of permutations $\pi:[N]\to[N]$. Let
$\pi\in S(N)$ be a permutation and let $W_\pi$ be the permutation unitary in $\cH_A$
induced by $\pi$. Such a unitary reorders the output systems according to $\pi$.
In this case, $\cH_A$ is decomposed to
\begin{align}
\cH= \bigoplus_{\lambda } {\cal U}_\lambda \ox {\cal V}_\lambda  ,
\end{align}
where ${\cal U}_\lambda$ is an irreducible space of the group $\unitary{\cH}$ and
${\cal V}_\lambda$ is an irreducible space of the permutation group $S(N)$.

When the group $G$ is chosen as $\unitary{\cH}$,
the multiplicity-free condition is not satisfied because
the dimension of ${\cal V}_\lambda$ shows the multiplicity of the representation
${\cal U}_\lambda $.
When the group $G$ is chosen as $S(N)$,
the multiplicity-free condition is not satisfied because
the dimension of ${\cal U}_\lambda$ shows the multiplicity of the representation
${\cal V}_\lambda $.
However, when the group $G$ is chosen as $\unitary{\cH}\times S(n)$,
the multiplicity-free condition is satisfied because
the spaces ${\cal U}_\lambda \ox {\cal V}_\lambda $
are different irreducible spaces.

In the following, to evaluate the capacities $C_{\EncG,\DecL}(\Psi_{AF})
,C_{\EncG,\DecOne}(\Psi_{AF}),C_{\EncG,\DecG}(\Psi_{AF})$,
we assume that $d=2$ and $\Psi_A=\rho^{\ox N}$ with a density matrix $\rho$ on $\cH=\mathbb{C}^2$.
Also, we assume that the eigenvalues of $\rho$ are $p$ and $1-p$ with $1\ge 2p$.
Then, $\Shannon(\Psi_A)=\Shannon(F)_\xi$ is $N h(p)$, where $h$ is the binary entropy.
In this case, the irreducible space is labeled by $k=0, \ldots, \lfloor N/2\rfloor$.
The dimension of the irreducible space $\cH_k$ is
$(N+1-k)({N \choose k}-{N \choose k-1})$, where ${N \choose -1}$ is defined to be $0$. Then we have
\begin{align}
  P_K(k)=\left({N \choose k}-{N \choose k-1}\right)q_k,
\end{align}
where $q_k:=\frac{p^k(1-p)^{N-k+1}-p^{N-k+1}(1-p)^k}{1-2p}$. Hence,
\begin{align}
\Shannon(K)_\xi &= - \sum_{k=0}^{\lfloor N/2\rfloor}
      q_k \left({N \choose k}-{N \choose k-1}\right) \log q_k \left({N \choose k}-{N \choose k-1}\right), \\
\Shannon(A)_\xi &= - \sum_{k=0}^{\lfloor N/2\rfloor}q_k
      \left({N \choose k}-{N \choose k-1}\right)\log \frac{q_k}{N+1-k}.
\end{align}
Therefore,
\begin{subequations}
\begin{align}
  C_{\EncG,\DecL}(\Psi_{AF})
&= - \sum_{k=0}^{\lfloor N/2\rfloor} q_k\left({N \choose k}-{N \choose k-1}\right)
        \log \frac{q_k}{N+1-k} -Nh(p),\Label{eq:L1} \\
  C_{\EncG,\DecOne}(\Psi_{AF})
&= - \sum_{k=0}^{\lfloor N/2\rfloor} q_k \left({N \choose k}-{N \choose k-1}\right)
        \log\frac{q_k}{N+1-k},\Label{eq:L2} \\
    C_{\EncG,\DecG}(\Psi_{AF})
&=  \Shannon(A)_\xi+\Shannon(F|K)_\xi = \Shannon(A)_\xi+\Shannon(F)-\Shannon(K)_\xi \\
&=  \sum_{k=0}^{\lfloor N/2\rfloor} q_k \left({N \choose k}-{N \choose k-1}\right)
      \log (N+1-k) \left({N \choose k}-{N \choose k-1}\right) + Nh(p).\Label{eq:L3}
\end{align}
\end{subequations}

We visualize these three dense coding capacities as functions of $N$ in Figure~\ref{fig:schur-duality}.
From the figure we can clearly see the dense coding power hierarchy of different decoders:
the less the locality constraint on the decoders, the larger the corresponding dense coding capacity.

\begin{figure}[!htbp]
  \centering
  \includegraphics[width=0.7\textwidth]{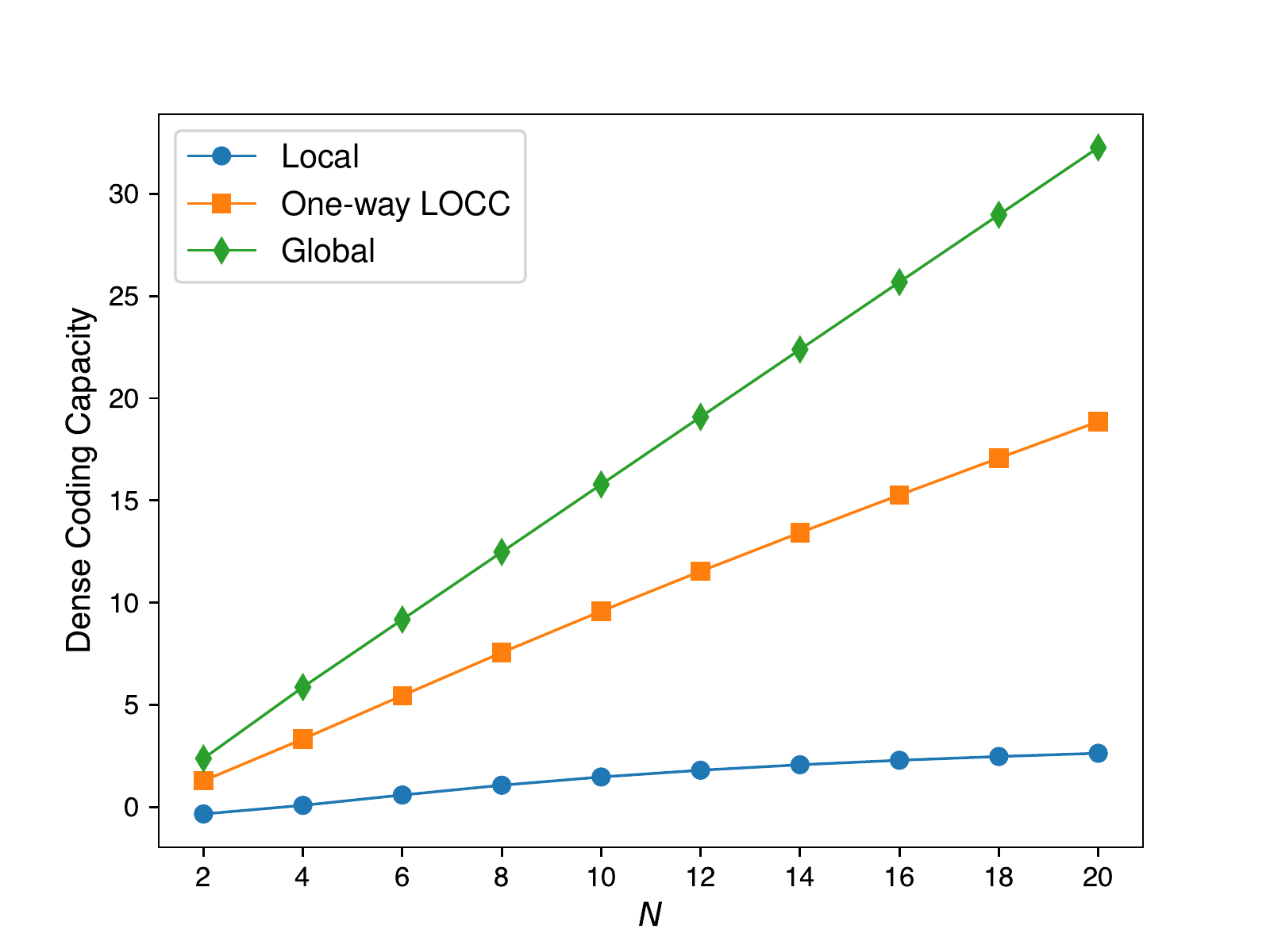}
  \caption{Three dense coding capacities---$C_{\EncG,\DecL}(\Psi_{AF})$ with local encoders,
  $C_{\EncG,\DecOne}(\Psi_{AF})$ with one-way LOCC decoders,
  and $C_{\EncG,\DecG}(\Psi_{AF})$ with global decoders---as functions of $N$,
  where $N$ is the number of identical parties. Parameter $p$ is set to $p=1/4$.}
  \Label{fig:schur-duality}
\end{figure}

\section{Extension to non-quantum preshared state}\Label{S8}

In the above dense coding framework, we assume that the preshared resource on the
bipartite system $\cH_{AF}$ is a bipartite quantum state (positive semidefinite operator with unit trace).
However, if the decoders are limited to separable measurements $\DecSEP$ or PPT measurements $\DecPPT$,
it is theoretically possible that the state on the bipartite system $\cH_{AF}$ is not a quantum state,
that is, we can loosen the positive semidefiniteness constraint.

As a demonstrative example, we assume that the available decoders are separable measurements $\DecSEP$.
We denote the cone composed of separable operators on the bipartite system $\cH_{AF}$ by $\SEP$ and
the dual cone by $\SEP^*$. Then, we define the set $\cS(\SEP^*):=\{ \rho \in \SEP^*| \tr \rho=1\}$.
For the preshared ``resource'' $\rho_{AF} \in \cS(\SEP^*)$, we can define analogously
the dense coding capacities
\begin{align}
  & C^\varepsilon_{\fE,\fD}(\rho_{AF}),\quad
    C_{\fE,\fD}(\rho_{AF}),\quad
    C^\dagger_{\fE,\fD}(\rho_{AF}), \\
  & \wt{C}^\varepsilon_{\fE,\DecOne}(\rho_{AF}),\quad
    \wt{C}_{\fE,\DecOne}(\rho_{AF}),\quad
    \wt{C}^\dagger_{\fE,\DecOne}(\rho_{AF}),
\end{align}
where $\fE\in\{\EncG,\EncCP,\EncPPT, \EncP\}$ and $\fD\in\{\DecOne,\DecLOCC,\DecSEP\}$
in the same way as Section~\ref{sec:task}.

Similarly, we denote the cone composed of PPT operators on the bipartite system $\cH_{AF}$
by $\PPT$ and the dual cone by $\PPT^*$.
Then, we define the set $\cS(\PPT^*):=\{ \rho \in \PPT^*| \tr \rho=1\}\subset \cS(\SEP^*)$.
For the preshared ``resource'' $\rho_{AF}^\prime \in \cS(\PPT^*)$, we can define in the same way
the dense coding capacities
\begin{align}
    C^\varepsilon_{\fE,\DecPPT}(\rho^\prime_{AF}),\quad
    C_{\fE,\DecPPT}(\rho_{AF}^\prime),\quad
    C^\dagger_{\fE,\DecPPT}(\rho_{AF}^\prime),\quad
\end{align}
where $\fE\in\{\EncG,\EncCP,\EncPPT\}$.

Regarding the above non-quantum preshared state extension,
we have the following strong converse theorem,
much like the strong converse parts of Theorems \ref{thm:enhanced-relation} and  \ref{thm:asymptotic characterization}.
See Appendix~\ref{appx:w-con2} for the proof.

\begin{theorem}[Strong converse part]\Label{thm:w-con2}
For $\rho_{AF} \in \cS(\SEP^*)$ and $\rho_{AF}^\prime \in \cS(\PPT^*)$, it holds that
\begin{align}
\wt{C}_{\EncP,\DecOne}^\dagger(\rho_{AF}) &\leq \Shannon\left(\cG(\rho_A)\right),\Label{eq:w-con2-1}\\
C_{\EncP,\DecSEP}^\dagger \left(\rho_{AF}\right) &\le \Shannon\left(\cG(\rho_A)\right),\Label{eq:w-con2-2}\\
C_{\EncPPT,\DecPPT}^\dagger \left(\rho_{AF}^\prime\right) &\le \Shannon\left(\cG(\rho_A^\prime)\right).\Label{eq:w-con2-3}
\end{align}
\end{theorem}

However, we are not able to prove the direct part as that of Theorem \ref{thm:asymptotic characterization}
for the the non-quantum preshared state extension expect that some additional conditions are satisfied.
The results are summarized in the following theorem. See Appendix~\ref{appx:w-con2} for the proof.

\begin{theorem}[Direct part]\Label{thm:asymptotic-direct2}
Assume that the (projective) unitary representation $U$ on $\cH_A$ satisfies
Assumption~\ref{assp:multiplicity-free} (the multiplicity-free condition).
Given $\rho_{AF} \in \cS(\SEP^*)$ and $\rho_{AF}' \in \cS(\PPT^*)$,
we choose purifications $ \Psi_{AF}$ and $ \Psi_{AF}$ of $\rho_A$ and $\rho_A'$, respectively.
It holds that
\begin{itemize}
  \item If there exists a trace-preserving positive operation $\cE_F\in\channel{F\to F}_{\rm p}$ such that $\cE_F(\rho_{AF})= \Psi_{AF}$, we have
        \begin{align}
         \Shannon\left(\cG(\rho_A)\right) \le C_{\EncG,\DecOne}(\rho_{AF})\Label{LS12}.
        \end{align}
  \item If there exists a trace-preserving operation $\cE_F'\in \channel{F\to F}_{\rm ppt}$ such that $\cE_F'(\rho_{AF}')= \Psi_{AF}'$, we have
        \begin{align}
         \Shannon\left(\cG(\rho_A')\right) \leq C_{\EncG,\DecOne}(\rho_{AF}')\Label{LS13}.
        \end{align}
\end{itemize}
\end{theorem}

\section{Conclusion}

In this paper, we have investigated thoroughly the classical communication
with various restriction on decoding operations
via quantum
resources task within the resource theory of asymmetry.
We have discussed the ultimate limit of the transmission rate when
the encoding operations are relaxed to the most general operations allowed in the framework of GPT
\MH{as super-quantum encoders}
and a certain locality condition is imposed to the decoding measurement.
In the analysis, we have employed a variant of the dense coding protocol by considering the resource theory of asymmetry.
In this variant, the preshared entangled state is fixed and
the decoding measurement is restricted to a local measurement.
When the group representation characterizing the resource theory of asymmetry
satisfies the multiplicity-free condition, we have proven that
this kind of relaxation does not improve the transmission rate.

Going one step further, we have shown that
the same conclusion holds whenever the initial state is not a quantum state but satisfies certain condition.
Many interesting problems remain open. First or all,
we have imposed a strong condition (the multiplicity-free condition in Assumption~\ref{assp:multiplicity-free})
when proving the direct part of Theorem \ref{thm:asymptotic characterization}.
It would be interesting to inspect this condition.
What's more, we can consider the case when a certain locality condition,
for example the separability condition, is imposed on the initial state
and a class of super-quantum measurement is allowed as the decoding measurement.
Under this condition, our encoder can be relaxed to the class $\EncPPT$ or $\EncP$.
It is challenging to clarify whether this relaxation can yield higher transmission rates or not.

As another future problem,
under the same constraints for encoders and decoders,
we can consider the case when
any pre-shared entangled state is allowed between Alice and Fred
and the channel between Alice and Bob is given as a noisy quantum channel.
It is interesting to clarify whether super-quantum encoder enhances the capacity even in this setting.

\section*{Acknowledgements}
MH was supported in part by Guangdong Provincial Key Laboratory (Grant No. 2019B121203002).
Part of this work was done when KW was at the Southern University of Science and Technology.

\bibliography{references.bib}


\appendices

\section{Proof of Proposition~\ref{prop:regularized-asymmetry-of-assistance}}
\label{appx:prop:regularized-asymmetry-of-assistance}

\begin{proof}
~The first inequality follows by definition. To show the second inequality,
note that $\Shannon(\cG(\rho))$ upper bounds $A_G\left(\rho\right)$ due to the concavity of quantum entropy:
\begin{align}
  A_G\left(\rho\right)
&= \max_{\rho=\sum_xp_X(x)\proj{\psi_x}}\sum_x p_X(x) \Shannon\left(\cG(\psi_x)\right) \\
&\leq \max_{\rho=\sum_xp_X(x)\proj{\psi_x}}\Shannon\left(\cG\left(\sum_xp_X(x)\psi_x\right)\right) \\
&= \Shannon(\cG(\rho)).
\end{align}
This yields
\begin{align}
    A_G^\infty\left(\rho\right)
:= \limsup_{n\to\infty}\frac{1}{n}A_G\left(\rho^{\ox n}\right)
\leq \lim_{n\to\infty}\frac{1}{n}\Shannon\left(\cG^{\ox n}\left(\rho^{\ox n}\right)\right)
= \Shannon(\cG(\rho)),
\Label{MN1}
\end{align}
where the last equality follows from that the quantum entropy is additive w.r.t. tensor product.
\end{proof}

\section{Proof of Theorem~\ref{thm:enhanced-relation}}\label{appx:thm:enhanced-relation}

\begin{proof}
Eq.~\eqref{eq:enhanced-relation-1} follows from definition.

\vspace*{0.05in}
The first inequality in Eq.~\eqref{eq:enhanced-relation-2} follows from definition.
Now we show the second inequality in Eq.~\eqref{eq:enhanced-relation-2},
which is commonly called the weak converse bound.
Consider an enhanced dense coding protocol depicted in Figure~\ref{fig:ea-dense-coding-enhanced}.
By definition, we must search exhaustively over all possible POVMs in the
environment to optimize the quantity $\wt{C}_{\fE,\DecOne}^\varepsilon(\Psi_{AF})$, which is notoriously difficult.
Luckily, it can be shown that Fred can restrict measurements to rank-one POVMs yet still achieve the
same information transmission performance~\cite{gregoratti2003quantum}. On the other hand, rank-one POVMs at
Fred's side are in one-to-one correspondence with pure state decompositions of $\Psi_A$ by the Sch\"{o}dinger-HJW
theorem~\cite{hughston1993complete,kirkpatrick2003schrodinger}:
\begin{align}\Label{eq:pure-state-decomposition}
    \Psi_A = \sum_xp_X(x)\proj{\psi_A^x},
\end{align}
where $p_X(x)$ a probability distribution and $\{\ket{\psi_A^x}\}_x$ a set of pure states (not necessarily orthonormal).
As a result, the task becomes how well Alice and Bob can encode classical information using the pure state ensemble
$\{p_X(x),\psi_A^x\}$ on average. For each conditional state $\psi_A^x$, Alice performs
a conditional encoding operation $\cE^{m\vert x}_{A\to A}\in\fE$. In the single-shot case,
the conditional mutual information between Alice's message and Bob' state is evaluated as
\begin{align}
&\;\sum_{x}P_X(x)
\frac{1}{|{\cal M}|}\sum_{m \in {\cal M}}
\Rel\left({\cal E}^m_{A \to A}( \psi_A^x)
\rel
\frac{1}{|{\cal M}|}\sum_{m' \in {\cal M}} {\cal E}^{m'}_{A \to A}( \psi_A^x)
\right) \\
\le&\; \sum_{x}P_X(x)
\frac{1}{|{\cal M}|}\sum_{m \in {\cal M}}
\Rel\left({\cal E}^m_{A \to A}( \psi_A^x)
\rel
\cG( \psi_A^x)
\right) \\
\stackrel{(a)}{=}&\; \sum_{x}P_X(x)
\frac{1}{|{\cal M}|}
\sum_{m \in {\cal M}}
\Rel\left({\cal E}^m_{A \to A}( \psi_A^x)
\rel
{\cal E}^m_{A \to A} \circ \cG( \psi_A^x)
\right) \\
\stackrel{(b)}{\leq}&\;
\sum_{x}P_X(x) \frac{1}{|{\cal M}|}\sum_{m \in {\cal M}}
\Rel\left( \psi_A^x
\rel
\cG( \psi_A^x)
\right) \\
=&\; \sum_{x}P_X(x)
\Rel\left( \psi_A^x
\rel
 \cG( \psi_A^x)
\right) \\
\le&\; A_G(\Psi_A) \\
\le&\; A_G^\infty(\Psi_A),
\end{align}
where $(a)$ follows from the fact that $\cE^{m\vert x}_{A\to A}\in\fE$
and $(b)$ follows from the data processing inequality for the relative entropy with
respect to the trace-preserving positive operations \cite[Theorem 1]{muller2017monotonicity}.
Hence, when the state $\Psi_A^{\ox n}$ is given,
the conditional mutual information between
Alice's message and Bob's state is upper bounded by
$n A_G^\infty(\Psi_A)$. Combining Fano’s inequality,
we can show the inequality $\wt{C}_{\fE,\DecOne}(\Psi_{AF}) \le A_G^\infty(\Psi_A)$.

\vspace*{0.05in}
The first inequality in Eq.~\eqref{eq:enhanced-relation-3} follows from definition.
Now we show the second inequality in Eq.~\eqref{eq:enhanced-relation-3},
which is commonly called the strong converse bound.
It can be shown by use of a meta-converse technique originally invented in~\cite{nagaoka2001strong} and
further investigated in~\cite[Chapter 3]{hayashi2005asymptotic} (see
also~\cite{sharma2013fundamental,wilde2014strong,wang2020permutation} for more applications of this technique). Roughly
speaking, this meta-converse method guarantees that a quantum divergence satisfying certain reasonable properties
induces an upper bound on the success probability of the communication protocol. Here we adopt the Petz \Renyi
divergence $\PRRel{\alpha}$~\cite{petz1986quasi}, which meets all required properties (cf. Section~\ref{sec:Quantum
entropies}). In the following Proposition~\ref{prop:one-shot-strong-converse} (which will be proved shortly),
we upper bound the success probability of any one-shot enhanced code $\cC\in(\fE,\DecOne)$ in terms of the \Renyi entropy,
then the strong converse bound follows by considering block coding.
Notice that $\lim_{\alpha\to1}\PRenyi{2-\alpha}(\cG(\Psi_A))=\Shannon(\cG(\Psi_A))$. What's more, $\PRenyi{\alpha}$
is continuous and monotonically decreasing in $\alpha$. Applying the standard argument outlined
in~\cite{nagaoka2001strong,sharma2013fundamental}, we obtain from Proposition~\ref{prop:one-shot-strong-converse} that
$\Shannon(\cG(\Psi_A))$ is actually a strong converse bound.
\end{proof}

\begin{proposition}\Label{prop:one-shot-strong-converse}
Let $\fE\in\{\EncG,\EncCP,\EncPPT,\EncP\}$.
Any enhanced dense coding code $\cC\in(\fE,\DecOne)$ as illustrated in Figure~\ref{fig:ea-dense-coding-enhanced}
obeys the following bound for arbitrary $\alpha\in(1,2)$:
\begin{align}\Label{eq:one-shot-strong-converse}
s(\cC) \leq \exp\left\{\frac{\alpha-1}{\alpha}\left(\PRenyi{2-\alpha}(\cG(\Psi_A)) - \log\vert\cC\vert\right)\right\},
\end{align}
where $\PRenyi{\alpha}$ is the \Renyi entropy defined in~\eqref{eq:renyi entropy}.
\end{proposition}
\begin{proof}
Let $\{\Lambda^x\}_{x\in\cX}$ be the measurement carried out by Fred.
Set $p_X(x) := \tr[\Lambda^x\Psi_F]$ and
$\psi_A^x :=\tr_F\left[\left(\1_A\ox\Lambda^x\right)\Psi_{AF}\right]/p_X(x)$.
We define the following two quantum states:
\begin{align}
    \rho_{MXA} &:= \frac{1}{\vert\cM\vert}\sum_{m} \proj{m}_M\ox\sum_xp_X(x)\proj{x}_X\ox\cE^{m\vert x}(\psi_A^x),\\
    \sigma_{MXA} &:= \pi_M\ox\sum_xp_X(x)\proj{x}_X\ox\cG(\psi^x_A),
\end{align}
where $\rho_{MXA}$ serves as a \emph{test state}.
For given $\psi_A^x $, we choose a pure state decomposition as
$\psi_A^x=\sum_{y} P_{Y|X}(y|x)\psi_A^{x,y}$.
Then, we have
\begin{align}
    \rho_{MXYA} &:= \frac{1}{\vert\cM\vert}\sum_{m} \proj{m}_M\ox\sum_{x,y}p_{XY}(x,y)\proj{x,y}_X\ox\cE^{m\vert x}(\psi_A^{x,y}),\\
    \sigma_{MXYA} &:= \pi_M\ox\sum_{x,y}p_{XY}(x,y)\proj{x,y}_{XY}\ox\cG(\psi^{x,y}_A),
\end{align}

Given a code $(\{\cE^{m|x}\},\{\Gamma^{\wh{m}\vert x}\})$ depending on $x$,
the positive operator
\begin{align}
T:=  \sum_{m} \sum_x \proj{m}_M\ox\proj{x}_X\ox \Gamma^{{m}\vert x}
\end{align}
satisfies
\begin{align}
\tr T \rho_{MXA}&=  \frac{1}{\vert\cM\vert}\sum_{m}p_{\wh{M}M}(m\vert m)
= s(\cC), \\
\tr T \sigma_{MXA}&=  \frac{1}{\vert\cM\vert}\sum_xp_X(x)\tr\left[\sum_m \Gamma^{m\vert x}\cG(\psi^x_A)\right]
\stackrel{(a)}{=} \frac{1}{\vert\cM\vert}\sum_xp_X(x)\tr\left[\cG(\psi^x_A)\right]
= \frac{1}{\vert\cM\vert},
\end{align}
where $(a)$ follows from the fact that for
each $x$, the conditional decoding operation $\{\Gamma^{m\vert x}\}_{m\in\cM}$ forms a POVM.
Applying the information processing inequality to the binary measurement
$\{ T,I-T \}$, we have
\begin{align}
s(\cC)^\alpha \cdot\left(\frac{1}{\vert\cC\vert}\right)^{1-\alpha}
+(1-s(\cC))^\alpha \cdot \left(1-\frac{1}{\vert\cC\vert}\right)^{1-\alpha}
&\leq e^{(\alpha-1)\PRRel{\alpha}\left(\rho_{MXA}\rel\sigma_{MXA}\right)}  .
\end{align}
Thus, we have the following:
\begin{align}
 s(\cC)^\alpha \cdot \vert\cC\vert^{\alpha-1}
&= s(\cC)^\alpha \cdot \left(\frac{1}{\vert\cC\vert}\right)^{1-\alpha}
\leq s(\cC)^\alpha \cdot \left(\frac{1}{\vert\cC\vert}\right)^{1-\alpha}
+(1-s(\cC))^\alpha \cdot \left(1-\frac{1}{\vert\cC\vert}\right)^{1-\alpha}
\\
&\leq e^{(\alpha-1)\PRRel{\alpha}\left(\rho_{MXA}\rel\sigma_{MXA}\right)}
\\
&\stackrel{(a)}{\leq}
 e^{(\alpha-1)\PRRel{\alpha}\left(\rho_{MXYA}\rel\sigma_{MXYA}\right)}
\\
&=\frac{1}{\vert\cM\vert}\sum_{m} \sum_{x,y}p_{XY}(x,y)
e^{(\alpha-1)\PRRel{\alpha}\left(\cE^{m\vert x}(\psi_A^x)\rel\cG(\psi^x_A)\right)} \\
&\stackrel{(b)}{=}
\frac{1}{\vert\cM\vert}\sum_{m}
\sum_{x,y}p_{XY}(x,y)\int_G
e^{(\alpha-1)\PRRel{\alpha}\left(U_g \cE^{m\vert x}(\psi_A^{x,y}) U_g^\dagger\rel
U_g \cG(\psi^{x,y}_A) U_g^\dagger\right)} \nu(d g)\\
&\stackrel{(c)}{=} \frac{1}{\vert\cM\vert}\sum_{m}\sum_{xy}p_{XY}(x,y) \int_G
        e^{(\alpha-1)\PRRel{\alpha}\left(U_g \cE^{m\vert x}(\psi_A^{x,y}) U_g^\dagger
        \rel \cG(\psi^{x,y}_A)\right)}\nu(d g) \\
&= \frac{1}{\vert\cM\vert}\sum_{m}
\sum_{x,y}p_{XY}(x,y)
\tr\left[ (U_g (\cE^{m\vert x}(\psi^{x,y}_A)) U_g^\dagger)^\alpha \cG(\psi^{x,y}_A)^{1-\alpha}\right] \nu(dg)\\
&\stackrel{(d)}{\le} \frac{1}{\vert\cM\vert}\sum_{m}
\sum_{x,y}p_{XY}(x,y)
\tr\left[ (U_g (\cE^{m\vert x}(\psi^{x,y}_A)) U_g^\dagger) \cG(\psi^{x,y}_A)^{1-\alpha}\right] \nu(dg)\\
&= \frac{1}{\vert\cM\vert}\sum_{m}
 \tr \sum_{x,y}p_{XY}(x,y) \left[  \Big( \int_G
                     U_g (\cE^{m\vert x}(\psi^{x,y}_A)) U_g^\dagger \Big) \nu(dg) \cG(\psi^{x,y}_A)^{1-\alpha}\right] \\
&= \frac{1}{\vert\cM\vert}\sum_{m}
 \tr \sum_{x,y}p_{XY}(x,y) \left[
                    \cG \circ \cE^{m\vert x}(\psi^{x,y}_A)  \cG(\psi^{x,y}_A)^{1-\alpha}\right] \\
&\stackrel{(e)}{=}  \tr\left[ \sum_{x,y}p_{XY}(x,y)\cG(\psi^{x,y}_A)^{2-\alpha}\right] \\
&\stackrel{(f)}{\leq}  \tr\left[\cG\left( \sum_{x,y}p_{XY}(x,y)\psi^{x,y}_A\right)^{2-\alpha}\right] \\
&=  e^{(\alpha-1)\PRenyi{2-\alpha}(\cG(\Psi_A))},
\end{align}
where $(a)$ follows from the information processing inequality for $\PRRel{\alpha}$,
$(b)$ follows from the fact that $\PRRel{\alpha}$ is invariant w.r.t. unitary channel,
$(c)$ follows from the definition
of $\cG$,
$(d)$ follows from the inequality
$x^\alpha \le x$ for $x \in [0,1]$ and $\alpha >1$,
$(e)$ follows from the equation
$ \cG \circ \cE^{m\vert x}= \cG$ since $\cE^{m\vert x}\in\fE$, and
$(f)$ follows from that $t\mapsto t^{2-\alpha}$ is operator concave when
$\alpha\in(1,2)$~\cite[Table 2.2]{tomamichel2015quantum}
\cite[Appendix A.4]{hayashi2016quantum}. Rearranging the above inequality leads
to~\eqref{eq:one-shot-strong-converse}.
\end{proof}

\section{Proof of Theorem~\ref{thm:asymptotic characterization}}\label{S4B}

Based on Proposition~\ref{prop:regularized-asymmetry-of-assistance}, Proposition~\ref{prop:enhanced-relation2},
and Theorem~\ref{thm:enhanced-relation},
to show Theorem~\ref{thm:asymptotic characterization} it suffices to show the following inequalities:
\begin{align}\label{eq:cSiXYsXiDO}
    C_{\EncPPT,\DecPPT}^\dagger\left(\Psi_{AF}\right),
    C_{\EncP,\DecSEP}^\dagger \left(\Psi_{AF}\right)
\leq \Shannon\left(\cG(\Psi_A)\right) \leq C_{\EncG,\DecOne}(\Psi_{AF}),
\end{align}
where
\begin{itemize}
  \item The second inequality of Eq.~\eqref{eq:cSiXYsXiDO} is known as the direct part
        (or the achievability part), meaning that there exist encoding operations from $\EncG$
        and one-way LOCC decoding measurements from $\DecOne$
        for which the rate $\Shannon\left(\cG(\Psi_A)\right)$ is achievable.
        We show this direct part in Appendix~\ref{sec:asymptotic-direct}.
        Notice that the proof of one-shot direct part remains as the most difficult part in this work.
  \item The first inequality of Eq.~\eqref{eq:cSiXYsXiDO} is known as the strong converse bound,
        meaning that $\Shannon\left(\cG(\Psi_A)\right)$ is an upper bound on all possible achievable coding
        rates even if coding rate and decoding error tradeoff are allowed.
        We show this strong converse part in Appendix~\ref{sec:asymptotic-converse}.
\end{itemize}

\subsection{Direct part under the multiplicity-free condition}\Label{sec:asymptotic-direct}

In this section we prove second inequality of Eq.~\eqref{eq:cSiXYsXiDO}
under Assumption~\ref{assp:multiplicity-free} (the multiplicity-free condition).
We do so by first presenting an one-shot direct part and then applying it to the asymptotic regime.

\subsubsection{One-shot direct part}\Label{sec:one-shot}

In the one-shot direct part we present an one-shot characterization
on $C^\varepsilon_{\EncG,\DecOne}(\Psi_{AF})$, the one-shot $\varepsilon$-dense coding capacity
where the available encoders are the unitary representations $\EncG$
and the available decoders are the one-way LOCC measurements $\DecOne$.
We begin with some notations.
W.l.o.g., the purification of $\Psi_A$ can be chosen with the form
\begin{align}
    \ket{\Psi}_{AF} := \frac{1}{\sqrt{N}}\sum_{n\in\cN}\ket{\psi_{A,n}}\ket{n}_F,
\end{align}
where $\cN$ is some alphabet, $N\equiv\vert\cN\vert\geq \rank(\Psi_A)$ is the size of the alphabet, $\{\ket{n}\}$ is an
orthonormal basis of $F$, and $\{\ket{\psi_{A,n}}\}$ is a set of pure states (not necessarily orthonormal) of $A$.
Also, $\rank(\Psi_A)$ expresses the rank of state $\Psi_A$.
In this purification, system $F$ is $N$-dimensional. We remark that such uniform purification is always possible as long
as $N\geq\rank(\Psi_A)$~\cite[Exercise 5.1.3]{nielsen2002introduction}. Under this purification, we have
$\Psi_A=\frac{1}{N}\sum_{n\in\cN}\proj{\psi_{A,n}}$. For each pure conditional state $\psi_{A,n}$, define its twirled
version as $\ol{\rho}_{A,n} :=
\cG(\psi_{A,n})$. Correspondingly, the twirled state of $\Psi_A$ is
\begin{align}
  \ol{\rho}_A := \cG(\Psi_A)
= \frac{1}{N}\sum_{n\in\cN}\cG(\psi_{A,n})
= \frac{1}{N}\sum_{n\in\cN}\ol{\rho}_{A,n}.
\end{align}
Then, the state $\xi_{AF} $ has another expression as
\begin{align}
    \xi_{AF} = \frac{1}{N}\sum_{n\in\cN} \ol{\rho}_{A,n} \ox \proj{n}_F.
\end{align}
Notice that $\xi_A=\ol{\rho}_A$.
Actually, $\xi_{AF}$ can be obtained from $\ket{\Psi}_{AF}$ by first dephasing
$F$ in the orthonormal basis and then twirling system $A$ via $\cG$.
We evaluate here various Petz-\Renyi entropies of $\xi_{AF}$ that are useful for later analysis:
\begin{subequations}\Label{eq:Petz-entropies}
\begin{align}
    \PRenyi{\alpha}(F\vert A)_\xi
    :=&\; - \PRRel{\alpha}\left(\xi_{AF}\rel\1_F\ox\xi_A\right)
     =    \frac{1}{1-\alpha}\log\frac{1}{N^\alpha}\sum_{n\in\cN}
          \tr\left[\ol{\rho}_{A,n}^\alpha\ol{\rho}_A^{1-\alpha}\right],\Label{eq:Petz-entropies-1} \\
      \PRenyi{\alpha}(A\vert F)_\xi
    :=&\; - \PRRel{\alpha}\left(\xi_{AY}\rel\1_A\ox\xi_F\right)
    = \frac{1}{1-\alpha}\log\tr\left[
      \frac{1}{N}\sum_{n\in\cN}\ol{\rho}_{A,n}^\alpha\right],\Label{eq:Petz-entropies-2} \\
    \PRenyi{\alpha}(A)_\xi :=&\; \frac{1}{1-\alpha}\log\tr \ol{\rho}_A^\alpha.\Label{eq:Petz-entropies-3}
\end{align}
\end{subequations}

We focus on two convex functions
$- s \PRenyi{1+s}(A F)_\xi +s \PRenyi{1-s}(F\vert A)_\xi$ and $- s \PRenyi{1+s}(A)_\xi )$.
The maximum of them, i.e.,
$\max(- s \PRenyi{1+s}(A F)_\xi +s \PRenyi{1-s}(F\vert A)_\xi,
- s \PRenyi{1+s}(A)_\xi )$ is also a convex function.
We define Legendre transformation of  the convex function as
\begin{align}
{\cal L}_{\xi}(R):=
\max_{0 \le s \le 1} s R+
\min \big(s \PRenyi{1+s}(A F)_\xi -s \PRenyi{1-s}(F\vert A)_\xi,
 s \PRenyi{1+s}(A)_\xi \big).\Label{Le1}
\end{align}

Now we are ready to state the one-shot direct coding theorem.

\begin{theorem}[One-shot direct part]\Label{thm:one-shot-characterization}
Let $\ket{\Psi}_{AF}$ be a bipartite pure quantum state and $\varepsilon\in[0,1)$.
When the (projective) unitary representation $U$ on $\cH_A$ satisfies
Assumption~\ref{assp:multiplicity-free} (the multiplicity-free condition), it holds that
\begin{align}\Label{eq:one-shot-characterization}
-{\cal L}_\xi^{-1}(-\log   \epsilon ) \leq C^\varepsilon_{\EncG,\DecOne}(\Psi_{AF}).
\end{align}
\end{theorem}

\begin{proof}
We prove Theorem~\ref{thm:one-shot-characterization} in the following four steps.

\vspace*{0.1in}{\bf Step 1:} We introduce a measurement induced by $2$-universal hash function.
Let $\cT\subset\cN$ be a strict subset of $\cN$ and set
$T:=\vert\cT\vert$. Let $f:\cN\to\cT$ be a linear surjective $2$-universal hashing function. The hashing function $f$
splits system $F$ into $T$ non-overlapping subspaces $\cS^{f,t}$ with the corresponding subspace projectors:
\begin{align}
    \Pi^{f,t} := \sum_{n\in f^{-1}(t)}\proj{n}_F.
\end{align}
Notice that $\{n\in f^{-1}(t)\}$ is of the same size for each $t$ and is given by $L:=N/T$. As
so, each subspace $\cS^{f,t}$ is $L$-dimensional. Fix the pair $(f,t)$. Let
$\bm{Q}^{f,t}:=\{Q_{A,n}^{f,t}\}_{n\in f^{-1}(t)}$ be a POVM on $A$.
Define the following $\delta$ function
\begin{align}
  \delta^{f,t} :=
  1 - \frac{1}{L}\sum_{n\in f^{-1}(t)}\tr\left[Q_{A,n}^{f,t}\ol{\rho}_{A,n}\right].
\end{align}
Roughly, $\delta^{f,t}$ quantifies how well the measurement $\bm{Q}^{f,t}$ detects the twirled states lying inside the
subspace projected by $\Pi^{f,t}$. Using the same discussion as the proof of
\cite[Theorem 7]{tomamichel2013hierarchy}, we are able to construct a list of POVMs $\bm{Q}^{f,t}$ such that $Q_{A,n}^{f,t}$ and
$\ol{\rho}_{A,n}$ are commutative for each $n$ and the expected value (w.r.t. both $f$ and $t$) of the $\delta$
function is upper bounded as follows
\begin{subequations}\Label{eq:delta-expectation}
\begin{align}
    \bE_{F,T}\delta^{F,T}
&=  1 - \bE_F \frac{1}{N}\sum_{n\in\cN}\tr\left[Q_{A,n}^{F,F(n)}\ol{\rho}_{A,n}\right] \\
&\leq \frac{1}{N}\sum_{n\in\cN}\tr\left[\ol{\rho}_{A,n}\{\ol{\rho}_{A,n} \geq L\ol{\rho}_A\}\right]
          + \tr\left[L\ol{\rho}_A\{\ol{\rho}_{A,n} < L\ol{\rho}_A\}\right] \\
&\stackrel{(b)}{\leq}\frac{L^{s'}}{N}\sum_{n\in\cN}\tr\left[ \ol{\rho}_{A,n}^{1-s'}\ol{\rho}_A^{s'}\right] \\
&\stackrel{(c)}{=}\frac{L^{s'}}{N} N^{1-s'}e^{s'\PRenyi{1-s'}(F\vert A)_\xi} \\
&\stackrel{(d)}{=} T^{-s'}e^{s'\PRenyi{1-s'}(F\vert A)_\xi},
\end{align}
\end{subequations}
where $s'\in[0,1]$, $(b)$ follows from~\cite[Theorem 1]{audenaert2007discriminating} (see also~\cite[Lemma
3.3]{hayashi2016quantum}) which is a well-known inequality in hypothesis testing, $(c)$ follows
from~\eqref{eq:Petz-entropies-1}, and $(d)$ follows from $L=N/T$. Based on the same construction, we can estimate the
following expectation w.r.t. both $F$ and $T$:
\begin{subequations}\Label{eq:H18}
\begin{align}
&\; \bE_{F,T}\tr\left[\left(\frac{1}{L}\sum_{n\in F^{-1}(T)}\ol{\rho}_{A,n}\right)^{1+s}\right] \\
=&\;  \bE_{F}\tr\left[\frac{1}{N} \sum_{n \in {\cal N}}
      \ol{\rho}_{A,n} \left(\frac{1}{L} \ol{\rho}_{A,n}
      +\frac{1}{L}\sum_{n'(\neq n)\in\cN:F(n')=F(n)}\ol{\rho}_{A,n'} \right)^s\right] \\
\stackrel{(a)}{\le} &\; \tr\left[\frac{1}{N} \sum_{n \in {\cal N}} \ol{\rho}_{A,n}
        \left(\frac{1}{L} \ol{\rho}_{A,n}
        + \bE_{F}\frac{1}{L}\sum_{n'(\neq n)\in\cN:F(n')=F(n)}\ol{\rho}_{A,n'} \right)^s\right] \\
=&\;   \tr\left[\frac{1}{N} \sum_{n \in {\cal N}} \ol{\rho}_{A,n}
        \left(\frac{1}{L}\ol{\rho}_{A,n} + \frac{L-1}{L}\ol{\rho}_A \right)^s\right] \\
\leq&\;   \tr\left[\frac{1}{N} \sum_{n \in {\cal N}} \ol{\rho}_{A,n}
        \left(\frac{1}{L}\ol{\rho}_{A,n} + \ol{\rho}_A \right)^s\right] \\
\stackrel{(b)}{\leq}&\; \tr\left[\frac{1}{N} \sum_{n \in {\cal N}} \ol{\rho}_{A,n}
        \left(\frac{1}{L^s}\ol{\rho}^s_{A,n} + \ol{\rho}^s_A \right)\right] \\
=&\;  \frac{1}{L^s}\tr\left[\frac{1}{N} \sum_{n \in {\cal N}} \ol{\rho}^{1+s}_{A,n}\right]
      + \tr\left[\ol{\rho}^{1+s}_A\right] \\
\stackrel{(c)}{=}&\;  \frac{1}{L^s}e^{-s \PRenyi{1+s}(A\vert F)_\xi} +e^{-s\PRenyi{1+s}(A)_{\xi}},
\end{align}
\end{subequations}
where $s\in[0,1]$, $(a)$ follows from the concavity of $x \mapsto x^{s}$ when $s\in(0,1]$~\cite[Table
2.2]{tomamichel2015quantum}\cite[Appendix A.4]{hayashi2016quantum}
and the Jensen inequality, $(b)$ follows from $(x+y)^s \le x^s +y^s$ for $x,y>0$,
and $(c)$ follows from~\eqref{eq:Petz-entropies-2} and~\eqref{eq:Petz-entropies-3}.

\vspace*{0.1in}{\bf Step 2:} We prepare a useful lemma that holds
under Assumption~\ref{assp:multiplicity-free} (the multiplicity-free condition) as follows.
This lemma will be shown in {\bf Step 4}.

\begin{lemma}\Label{lemma:one-shot-achievability}
Assume that the (projective) unitary representation $U$ on $\cH_A$ is multiplicity-free.
Then, the Hilbert space $\cH_A$ is decomposed as $\oplus_{k \in \cK} \cH_k$.
Let $\{Q_{A,l}\}_{l=1}^L$ be a POVM on $\cH_A$ such that each $Q_{A,l}$ is a projection onto
the invariant subspace $\oplus_{k\in\cS_l}\cH_k$,
where $\{\cS_l\}_{l=1}^L$ are disjoint subsets of $\cK$.
That is, each $Q_{A,l}$ projects into disjoint
irreducible subspaces $\cH_k$. Define the corresponding conditional projection operator in $AF$ as
\begin{align}
    Q_{AF} := \sum_{l=1}^L Q_{A,l}\ox\proj{l}_F,\Label{H1}
\end{align}
where $\{\ket{l}\}_{l=1}^L$ is an orthonormal basis of $F$.
Assume Alice and Fred preshare the pure bipartite state
\begin{align}
    \ket{\psi_{AF}} := \frac{1}{\sqrt{L}}\sum_{l=1}^L\vert\psi_{A,l}\rangle\ox\ket{l}_F,\Label{H2}
\end{align}
and define $\ol{\sigma}_A$ as
\begin{align}\Label{eq:sigma-A}
  \ol{\sigma}_A := \frac{1}{L}\sum_{l=1}^L Q_{A,l}\cG(\psi_{A,l})Q_{A,l}.
\end{align}
We have the following result regarding the one-way LOCC classical communication capability
of $\ket{\psi_{AF}}$.

Let $M$ be the message size. Let $E\equiv(\mathtt{g}_1,\cdots,\mathtt{g}_M)$ be a random coding such that each codeword
$\mathtt{g}_m$ is chosen independently and uniformly from $G$. We use the typewriter font $\mathtt{g}$ to indicate that
it is a random variable.
For a chosen encoder $(g_1,\cdots,g_M)$, there exists
a one-way LOCC decoder such
that the resulting protocol $\cC$'s expected decoding error is upper bounded as
\begin{align}\Label{eq:nvuveMhromojcIJt}
    \bE_{E}e(\cC(E)) \leq \frac{8M^s}{(1-\delta)^s}\tr\ol{\sigma}_A^{1+s} + 2\delta,
\end{align}
where $s\in(0,1)$, $\delta$ is defined as
\begin{align}\Label{eq:delta}
    \delta := 1 - \bra{\psi_{AF}}Q_{AF}\ket{\psi_{AF}}
= 1 - \frac{1}{L}\sum_{l=1}^L\langle\psi_{A,l}\vert Q_{A,l}\vert\psi_{A,l}\rangle
= 1 - \frac{1}{L}\sum_{l=1}^L\cG(\psi_{A,l}),
\end{align}
where the second equation follows from \eqref{H1} and \eqref{H2},
and the final equation follows from the invariant property of $Q_{A,l} $ for the group action.
\end{lemma}

\vspace*{0.1in}{\bf Step 3:} We proof Theorem \ref{thm:one-shot-characterization}
by applying Lemma \ref{lemma:one-shot-achievability} to the measurement induced by $2$-universal hash function.
That is,
we adopt the random coding argument and show that the expected decoding error for protocols generated by randomly
selecting codewords according to the uniform distribution and measurements according to $2$-universal hash function is
upper bounded.

In the first stage, Fred performs the projective measurement $\{\Pi^{f,t}\}$, diving system $F$ into $T$
subspaces. When the outcome is $t$, the post-measurement state on $AF$ is
\begin{align}
  \vert\Psi^{f,t}_{AF}\rangle := \frac{1}{\sqrt{L}}\sum_{n\in f^{-1}(t)}\ket{\psi_{A,n}}\ket{n}_F.
\end{align}
The outcome is communicated to Bob via a classical noiseless channel. Set $\bm{Q}_{AF}^{f,t}:=\sum_{n\in
f^{-1}(t)}Q_{A,n}^{f,t}\ox\proj{n}_F$, which is the conditional version of $\bm{Q}^{f,t}$. One can check that
\begin{align}
  \langle \Psi^{f,t}_{AF} \vert \bm{Q}_{AF}^{f,t} \vert \Psi^{f,t}_{AF} \rangle
= \frac{1}{L}\sum_{n\in f^{-1}(t)}
  \langle \psi_{A,n} \vert Q_{A,n}^{f,t} \vert \psi_{A,n} \rangle
= 1 - \delta^{f,t}.
\end{align}

In the second stage, we apply Lemma~\ref{lemma:one-shot-achievability} to the post-measurement state $\vert\Psi^{f,t}_{AF}\rangle$
with corresponding measurement $\bm{Q}^{f,t}$ to implement classical communication from Alice to Bob. We can do so
because the measurement $\bm{Q}^{f,t}$ satisfies the prerequisite given in Lemma~\ref{lemma:one-shot-achievability}. Consequently,
there exists a communication protocol with one-way LOCC decoder $\cC(e,f,t)$
depending on both the encoder $E$ and the
$2$-universal hash function $F$ and $T$ that satisfies the following decoding error condition
\begin{align}\Label{eq:MAUSAXnoufRw}
      \bE_E \epsilon(\cC(E,f,t))
\leq  \frac{8M^s}{(1-\delta^{f,t})^s}
      \tr\left[\left(\frac{1}{L}\sum_{n\in f^{-1}(t)}Q_{A,n}^{f,t}\ol{\rho}_{A,n}Q_{A,n}^{f,t}\right)^{1+s}\right]
      + 2\delta^{f,t}.
\end{align}

Averaging over all possible randomly generated codewords $E$ according to the uniform distribution and randomly
generated measurements according to $2$-universal hash functions $(F,T)$, we can upper bound the expected value of the
decoding error as follows:
\begin{subequations}\Label{eq:uGx}
\begin{align}
&\; \bE_{E,F,T}\epsilon(\cC(E,F,T)) \\
=&\;    \Pr\left(\delta^{F,T} \geq \frac{1}{2}\right)
        \bE_{E,F,T\vert \delta^{F,T}\geq 1/2}\epsilon(\cC(E,F,T))
      + \Pr\left(\delta^{F,T} < \frac{1}{2}\right)
        \bE_{E,F,T\vert \delta^{F,T} < 1/2}\epsilon(\cC(E,F,T)) \Label{eq:uGx1} \\
\stackrel{(a)}{\leq}&\; \Pr\left(\delta^{F,T} \geq \frac{1}{2}\right)
      + \Pr\left(\delta^{F,T} < \frac{1}{2}\right)
        \bE_{E,F,T\vert \delta^{F,T} < 1/2}\epsilon(\cC(E,F,T))  \\
\stackrel{(b)}{\leq}&\; 2\bE_{F,T}\delta^{F,T}
      + \Pr\left(\delta^{F,T} < \frac{1}{2}\right)
        \bE_{E,F,T\vert \delta^{F,T} < 1/2}\epsilon(\cC(E,F,T))  \\
\stackrel{(c)}{\leq}&\; 4\bE_{F,T}\delta^{F,T} \\
&\quad + \Pr\left(\delta^{F,T} < \frac{1}{2}\right)
    \bE_{F,T\vert \delta^{F,T} < 1/2}
    \left\{\frac{8M^s}{(1-\delta^{F,T})^s}
      \tr\left[\left(\frac{1}{L}\sum_{n\in F^{-1}(T)}Q_{A,n}^{F,T}\ol{\rho}_{A,n}Q_{A,n}^{F,T}\right)^{1+s}\right]
    \right\} \\
\stackrel{(d)}{\leq}&\; 4\bE_{F,T}\delta^{F,T}
 +  2^{s+3}M^s\bE_{F,T}
    \tr\left[\left(\frac{1}{L}\sum_{n\in F^{-1}(T)}Q_{A,n}^{F,T}\ol{\rho}_{A,n}Q_{A,n}^{F,T}\right)^{1+s}\right]  \\
\stackrel{(e)}{\leq}&\; 4\bE_{F,T}\delta^{F,T}
 +  2^{s+3}M^s\bE_{F,T}\tr\left[\left(\frac{1}{L}\sum_{n\in F^{-1}(T)}\ol{\rho}_{A,n}\right)^{1+s}\right]  \\
\stackrel{(f)}{\leq} &\;
    4 T^{-s'}e^{s'\PRenyi{1-s'}(F\vert A)_\xi}
+ 2^{s+3}\frac{M^s}{L^s}e^{-s \PRenyi{1+s}(A\vert F)_\xi} + 2^{s+3}M^se^{-s\PRenyi{1+s}(A)_{\xi}} \\
=&\;
    4 e^{-s'\left[\log T - \PRenyi{1-s'}(F\vert A)_\xi\right]}
  + 2^{s+3}e^{-s\left[\PRenyi{1+s}(A\vert F)_\xi + \log N - \log T - \log M\right]}
  + 2^{s+3}e^{-s\left[\PRenyi{1+s}(A)_{\xi} - \log M\right]}\\
\stackrel{(g)}{=} &\;
    4 e^{-s'\left[\log T - \PRenyi{1-s'}(F\vert A)_\xi\right]}
  + 2^{s+3}e^{-s\left[\PRenyi{1+s}(A F)_\xi  - \log T - \log M\right]}
  + 2^{s+3}e^{-s\left[\PRenyi{1+s}(A)_{\xi} - \log M\right]},\Label{eq:uGx2}
\end{align}
\end{subequations}
where
\begin{itemize}
  \item $\bE_{X\vert B}$ expresses the conditional expectation with respect to the variable $X$ conditioned on $B$,
  \item $(a)$ follows from the fact that the decoding error is less than $1$,
  \item $(b)$ follows from the Markov inequality~\cite[Eq. (3.31)]{cover2012elements}
        that $\Pr\left(\delta^{F,T} \geq 1/2\right) \leq 2\bE_{F,T}\delta^{F,T}$,
  \item $(c)$ follows from~\eqref{eq:MAUSAXnoufRw},
  \item $(d)$ follows from the relation that $\delta^{F,T}<1/2$ implies $(1-\delta^{F,T})^{-s}<2^s$.
        Notice that this relation is the essential reason why we divide the expectation into two regions --
        $\delta^{F,T}\geq1/2$ and $\delta^{F,T}<1/2$ -- in~\eqref{eq:uGx1}, since otherwise we cannot bound the term
        $(1-\delta^{F,T})^{-s}$,
  \item $(e)$ follows from the fact that the measurement element satisfies $0\leq Q_{A,n}^{f,t}\leq\1$ and
        thus the mutual commutativity property guarantees that $Q_{A,n}^{f,t} \ol{\rho}_{A,y} Q_{A,n}^{f,t} \leq
        \ol{\rho}_{A,n}$, and
  \item $(f)$ follows from the expectation estimations in~\eqref{eq:delta-expectation} and~\eqref{eq:H18}
        w.r.t. the $2$-universal hash function.
  \item $(g)$ follows from the fact that
  $\xi_F$ is the completely mixed state.
\end{itemize}

We set $s' = s$ in~\eqref{eq:uGx2} and solve the equation w.r.t. the variable $T$
\begin{align}
    \log T - \PRenyi{1-s}(F\vert A)_\xi = \PRenyi{1+s}(A F)_\xi - \log T - \log M,
\end{align}
yielding
\begin{align}
    \log T = \frac{1}{2}\left(\PRenyi{1-s}(F\vert A)_\xi + \PRenyi{1+s}(A F)_\xi - \log M\right).
\end{align}
Based on these choices, we can conclude from~\eqref{eq:uGx} that there exists a concrete communication protocol one-way LOCC decoder
$\cC(e,f,t)$ for carefully chosen encoding $e$ and the $2$-universal hash function $f$
such that its decoding error is upper bounded for $s \in [0,1]$ as follows
\begin{align}
    \epsilon\left(\cC(e,f,t)\right)
\leq & (4 + 2^{s+3})e^{-\frac{s}{2}\left[\PRenyi{1+s}(A F)_\xi - \PRenyi{1-s}(F\vert A)_\xi
      - \log M\right]}
  + 2^{s+3}e^{-s\left[\PRenyi{1+s}(A)_{\xi} - \log M\right]}\nonumber  \\
\leq & (4 + 16)e^{-\frac{s}{2}\left[\PRenyi{1+s}(A F)_\xi - \PRenyi{1-s}(F\vert A)_\xi
      - \log M\right]}
  +  16 e^{-s\left[\PRenyi{1+s}(A)_{\xi} - \log M\right]}.
\end{align}
Since $    \epsilon\left(\cC(e,f,t)\right)\le 1$, we have
\begin{align}
    \epsilon\left(\cC(e,f,t)\right)
\leq & 36
\min(1,
\max (e^{-\frac{s}{2}\left[\PRenyi{1+s}(A F)_\xi - \PRenyi{1-s}(F\vert A)_\xi - \log M\right]},
e^{-s\left[\PRenyi{1+s}(A)_{\xi} - \log M\right]})). \\
\le & 36
\min(1,
\max (e^{-\frac{s}{2}\left[\PRenyi{1+s}(A F)_\xi - \PRenyi{1-s}(F\vert A)_\xi - \log M\right]},
e^{-\frac{s}{2}\left[\PRenyi{1+s}(A)_{\xi} - \log M\right]})). \\
= & 36
\min(1, e^{-
\min(
\frac{s}{2}\left[\PRenyi{1+s}(A F)_\xi - \PRenyi{1-s}(F\vert A)_\xi - \log M\right],
\frac{s}{2}\left[\PRenyi{1+s}(A)_{\xi} - \log M\right])}). \\
= & 36
\min(1, e^{-\frac{1}{2} {\cal L}_\xi(-\log M)}).
\end{align}
That is,
\begin{align}
- 2 \log \frac{    \epsilon \left(\cC(e,f,t)\right)}{36}
\ge {\cal L}_\xi(- \log M).
\end{align}
Since ${\cal L}$ is monotonically increasing, we have
\begin{align}
-{\cal L}_\xi^{-1}(- 2 \log \frac{    \epsilon \left(\cC(e,f,t)\right)}{36})
\le  \log M.
\end{align}
This concludes the proof of Theorem \ref{thm:one-shot-characterization}. 

\vspace*{0.1in}{\bf Step 4:} Now we show Lemma \ref{lemma:one-shot-achievability}.
Define the $L$-th root of unity $\zeta:=\exp(2\pi i/L)$. From $\{\ket{l}_F\}_{l=1}^L$ we construct the induced Fourier
basis measurement $\{\vert\bm{b}_F^{l'}\rangle\}_{l'=1}^L$ via
\begin{align}
  \vert\bm{b}_F^{l'}\rangle
:= \frac{1}{\sqrt{L}}\sum_{l=1}^L\zeta^{ll'}\ket{l}_F,\; l'=1,\cdots,L.
\end{align}
For each $l'=1,\cdots, L$, defined the following subnormalized pure quantum state
\begin{align}
  \vert\phi_{A,l'}\rangle
:= \langle\bm{b}_F^{l'}\vert Q_{AF}\ket{\psi_{AF}}
= \frac{1}{L}\sum_{l=1}^L\zeta^{-ll'}Q_{A,l}\vert\psi_{A,l}\rangle,
\Label{eq:phi-A-l}
\end{align}
whose norm can be calculated as
\begin{align}\Label{eq:norm}
  \langle\phi_{A,l}\vert\phi_{A,l}\rangle
= \frac{1}{L^2}\sum_{l=1}^L\langle\psi_{A,l}\vert Q_{A,l}\vert\psi_{A,l}\rangle
= \frac{1-\delta}{L},
\end{align}
where the last equality follows from the definition of $\delta$~\eqref{eq:delta}.

Fred now perform this Fourier basis
measurement on $\ket{\psi_{AF}}$. After measurement, Fred holds the classical outcome $l$ and Alice holds the
post-measurement pure state. This leads to the classical-quantum state
\begin{align}\Label{eq:yyvPKart1}
  \sigma_{AF}
:= \sum_{l=1}^L\braket{\bm{b}_F^l}{\psi}_{AF}\braket{\psi}{\bm{b}_F^l}\ox\proj{l}_F.
\end{align}
Applying the pinching lemma~\cite[Lemma 3.10]{hayashi2016quantum} to the quantum state $\vert\psi_{AF}\rangle$ and the
binary projective measurement $\{Q_{AF},\1_{AF}-Q_ {AF}\}$ gives
\begin{align}\Label{eq:yyvPKart2}
  \proj{\psi}_{AF} \leq 2Q_{AF}\proj{\psi}_{AF} Q_{AF} + 2(\1_{AF}-Q_{AF})\proj{\psi}_{AF}(\1_{AF}-Q_{AF}).
\end{align}
Substituting~\eqref{eq:yyvPKart2} into~\eqref{eq:yyvPKart1} yields the following inequality
regarding $\sigma_{AF}$:
\begin{align}
  \sigma_{AF}
&\leq \sum_{l=1}^L\langle\bm{b}_F^l\vert\left(2Q\proj{\psi} Q + 2(\1-Q)\proj{\psi}(\1-Q)\right)
        \vert\bm{b}_F^l\rangle\ox\proj{l}_F \\
&=  2\sum_{l=1}^L\proj{\phi_{A,l}}\ox\proj{l}_F
    + 2\sum_{l=1}^L\langle\bm{b}_F^l\vert(\1-Q_{AF})\proj{\psi}(\1-Q_{AF})\vert\bm{b}_F^l\rangle\ox\proj{l}_F.
    \Label{eq:MpfOTC}
\end{align}

Define the following classical-quantum channel
\begin{align}\Label{eq:WCvrUeBTtveLxo}
    g \mapsto W^g_{BF} :=
    \frac{1}{1-\delta}\sum_{l=1}^{L} U_g\proj{\phi_{A,l}} U_g^\dagger\ox\proj{l}_F,
\end{align}
with uniform probability distribution on $G$.
One can verify that such defined $W^g_{BF}$ is indeed a quantum state using~\eqref{eq:norm}.
This classical-quantum channel has the following averaged state
\begin{align}
&\int_G   \Big(\sum_l\left(\frac{1}{1-\delta}U_g\proj{\phi_{A,l}} U_g^\dagger\right)
         \ox \proj{l}_F\Big) \nu(dg)\nonumber \\
=&\; \int_G W^g_{BF} \nu(dg)\nonumber \\
\stackrel{(a)}{=}&\; \frac{1}{(1-\delta) L^2}\int_G\sum_{l=1}^{L}
\sum_{l'=1}^{L} \sum_{l''=1}^{L}
\zeta^{ l (l'-l'') }
      U_g Q^{l''} |\psi_{A,l''} \rangle \langle \psi_{A,l'}| Q^{l'} U_g^\dagger\ox\proj{l}_F \nu(dg) \nonumber \\
\stackrel{(b)}{=}&\; \frac{1}{(1-\delta) L^2}\int_G \sum_{l=1}^{L}\sum_{l'=1}^{L}
      U_g Q^{l'}\proj{\psi_{A,l'}}Q^{l'} U_g^\dagger\ox\proj{l}_F \nu(dg) \nonumber \\
\stackrel{(c)}{=}&\;  \frac{1}{(1-\delta)L^2}\sum_{l=1}^{L}\sum_{l'=1}^{L}
      Q^{l'}\cG( \proj{\psi_{A,l'}} ) Q^{l'}\ox\proj{l}_F   \nonumber\\
\stackrel{(d)}{=}&\; \frac{1}{1-\delta}\ol{\sigma}_A\ox\pi_F,
\Label{eq:averaged state}
\end{align}
where $(a)$ follows \eqref{eq:phi-A-l},
$(b)$ follows because $Q_{A,l}$ is a projection onto disjoint $\oplus_{k\in\cS_l}\cH_k$ and the irreducible
representations on $\cH_k$ are not equivalent.
The multiplicity-free condition is used in this step.
In addition, we remark that the classical-quantum channel~\eqref{eq:WCvrUeBTtveLxo}
can be obtained by Fred performs the Fourier basis measurement and Alice adopts random coding chosen uniformly from
$G$.
$(c)$ follows from the commutativity between $U_g$ and $ Q^{l'}$ and the
commutativity follows from the invariance property.
$(d)$ follows from \eqref{eq:sigma-A}.

Applying the direct part of the classical-quantum channel coding theorem~\cite[Lemma 4.6]{hayashi2016quantum}, we
conclude that there exists an encoder $(g_1,\cdots,g_M)$ and a decoder $\bm{\Gamma}\equiv\{\Gamma^m\}_{m=1}^M$ in
$BF$ as a one-way LOCC measurement from Fred to Bob such
that the decoding error is upper bounded for arbitrary $s\in[0,1]$:
\begin{subequations}\Label{eq:WCvrUeBTtveLxo2}
\begin{align}
\text{decoding error}:=&\; \frac{1}{M}\sum_{m=1}^M\tr\left[W^{g_m}\left(\1 - \Gamma^m\right)\right] \\
\stackrel{(a)}{\leq}&\; {4M^s}\int_G\tr\left[(W^g)^{1-s}\ol{W}^s\right]\nu(dg) \\
=&\; 4M^s\int_G
    \tr\left[\left(\sum_{l=1}^{L} \frac{L}{1-\delta}
    U_g\proj{\phi_{A,l}} U_g^\dagger\ox\frac{1}{L}\proj{l}_F\right)^{1-s}
      \left(\frac{1}{1-\delta}\ol{\sigma}_A\ox\pi_F\right)^s\right] \nu(dg)\\
=&\;  4M^s \int_G
      \tr\left[\sum_l\left(\frac{L}{1-\delta}U_g\proj{\psi_{A,l}} U_g^\dagger\right)^{1-s}
      \left(\frac{1}{1-\delta}\ol{\sigma}_A\right)^s \ox \frac{1}{L}\proj{l}_F\right] \nu(dg)\\
\stackrel{(b)}{=}&\; 4M^s \int_G
      \tr\left[\sum_l\left(\frac{L}{1-\delta}U_g\proj{\phi_{A,l}} U_g^\dagger\right)
        \left(\frac{1}{1-\delta}\ol{\sigma}_A\right)^s \ox \frac{1}{L}\proj{l}_F\right] \nu(dg)\\
=&\; 4M^s \tr\left[ \int_G        \Big(\sum_l\left(\frac{1}{1-\delta}U_g\proj{\phi_{A,l}} U_g^\dagger\right)
         \ox \proj{l}_F\Big) \left( \Big(\frac{1}{1-\delta}\ol{\sigma}_A\Big)^s
        \otimes I_F \right)\nu(dg) \right]\\
\stackrel{(c)}{=}&\; \frac{4M^s}{(1-\delta)}\tr\left[ \Big(\sum_l
        \ol{\sigma}_A \ox \frac{1}{L}\proj{l}_F\Big)  \left( \Big(\frac{1}{1-\delta}\ol{\sigma}_A\Big)^s
        \otimes I_F \right)\right] \\
=&\; \frac{4M^s}{(1-\delta)^{1+s}}\tr\left[\sum_l
        \ol{\sigma}_A^{1+s} \ox \frac{1}{L}\proj{l}_F\right] \\
=&\;  \frac{4M^s}{(1-\delta)^{1+s}}\tr\left[\ol{\sigma}_A^{1+s}\right],
\end{align}
\end{subequations}
where $(a)$ follows from~\cite[Lemma 4.1]{hayashi2016quantum}, $(b)$ follows from~\eqref{eq:norm} implying that
$\frac{L} {1-\delta}U_g\proj{\phi}_{A,l} U_g^\dagger$ is a normalized pure state,
and $(c)$ follows from \eqref{eq:averaged state}.

\vspace*{0.1in}

Now we are ready to give a concrete protocol $\cC$ achieving the decoding error concluded
in~\eqref{eq:nvuveMhromojcIJt}. In this protocol, Fred adopts the Fourier basis measurement
$\{\vert\bm{b}_F^l\rangle\}_{l=1}^L$, Alice adopts the encoding $(g_1,\cdots,g_M)$, and Bob adopts the decoder
$\bm{\Gamma}$ originally designed for the classical-quantum channel $g\mapsto W_{AF}^g$. This is an communication protocol with one-way LOCC decoder
for $\ket{\psi_{AF}}$ since $\bm{\Gamma}$ is essentially an one-way LOCC
decoder from Fred to Bob. Thanks to the above analysis, we can evaluate the decoding error of $\cC$ as follows:
\begin{subequations}
\begin{align}
e(\cC) :=&\; \frac{1}{M}\sum_{m=1}^M \tr\left[\left(U_{g_m}\sigma_{AF}U^\dagger_{g_m}\right)
            \left(\1-\Gamma^m\right)\right] \\
\stackrel{(a)}{\leq}&\; \frac{1}{M}\sum_{m=1}^M\tr\left[U_{g_m}
      \left(2\sum_{l=1}^L\proj{\phi_{A,l}}\ox\proj{l}_F\right)
      U_{g_m}^\dagger\left(\1 - \Gamma^m\right)\right] \\
&\quad+ \frac{1}{M}\sum_{m=1}^M\tr\left[U_{g_m}
      \left(2\sum_{l=1}^L\langle\bm{b}^l\vert(\1-Q)\proj{\psi}(\1-Q)\vert\bm{b}^l\rangle\ox\proj{l}_F\right)
      U_{g_m}^\dagger\left(\1 - \Gamma^m\right)\right] \\
\stackrel{(b)}{=}&\;    \frac{2(1-\delta)}{M}\sum_{m=1}^M\tr\left[W^{g_m}\left(\1 - \Gamma^m\right)\right] \\
 &\quad+ \frac{1}{M}\sum_{m=1}^M\tr\left[U_{g_m}
      \left(2\sum_{l=1}^L\langle\bm{b}^l\vert(\1-Q)\proj{\psi}(\1-Q)\vert\bm{b}^l\rangle\ox\proj{l}_F\right)
      U_{g_m}^\dagger\left(\1 - \Gamma^m\right)\right] \\
\stackrel{(c)}{\leq}&\; \frac{8M^s}{(1-\delta)^s}\tr\left[\ol{\sigma}_A^{1+s}\right]
      + \frac{2}{M}\sum_{m=1}^M\tr\left[
        \sum_{l=1}^L\langle\bm{b}^l\vert(\1-Q)\proj{\psi}(\1-Q)\vert\bm{b}^l\rangle\ox\proj{l}_F\right] \\
\stackrel{(d)}{\leq}&\; \frac{8M^s}{(1-\delta)^s}\tr\left[\ol{\sigma}_A^{1+s}\right] + 2\delta,
\end{align}
\end{subequations}
where $(a)$ follows from~\eqref{eq:MpfOTC}, $(b)$ follows from the definition of
$W^g_{BF}$~\eqref{eq:WCvrUeBTtveLxo}, $(c)$ follows from~\eqref{eq:WCvrUeBTtveLxo2} and $\Gamma^m\geq0$,
and $(d)$ follows from that $\{\vert\bm{b}_F^l\rangle\}$ forms an orthonormal basis of $F$ and the definition
of $\delta$~\eqref{eq:delta}. We are done.
\end{proof}

\begin{remark}
Actually, our achievability proof (Theorem~\ref{thm:one-shot-characterization} and
Lemma~\ref{lemma:one-shot-achievability}) is inspired by the proof for~\cite [Theorem 1]{smolin2005entanglement}, which
we refer to as {SVW}. In the following, we compare in detail the similarity and uniqueness between our proof and {SVW}.
In general, both proofs are composed of two parts. In the first part, we apply a surjective linear hash function. This
mimics choosing the typical subspaces in {SVW}. In the second part, Fred performs an measurement on the Fourier basis.
This is the one-shot correspondence to the Fourier basis measurement in {SVW}. However, our task is different from the
task that is considered in {SVW}. As so, we need to invent different operations for both the sender and the receiver,
and manage a different evaluation method for the decoding error probability. On the other hand, {SVW} does not assume
the uniform distribution on the codewords \emph{a priori}. However, we do have this assumption due to the special
structure of the task under consideration. This uniformity assumption renders a more complicated proof so that it
becomes more difficult to derive an exponential upper bound.
\end{remark}

\subsubsection{Asymptotic direct part}

Based on the one-shot direct part in Theorem~\ref{thm:one-shot-characterization},
we are able to show the following coding theorem,
which concludes the second inequality of Eq.~\eqref{eq:cSiXYsXiDO}.

\begin{theorem}[Direct part]\Label{thm:asymptotic-direct}
Let $\ket{\Psi}_{AF}$ be a bipartite pure quantum state.
When the (projective) unitary representation $U$ on $\cH_A$ satisfies
Assumption~\ref{assp:multiplicity-free} (the multiplicity-free condition), it holds that
\begin{align}
  \Shannon\left(\cG(\Psi_A)\right) \leq C_{\EncG,\DecOne}(\Psi_{AF})\Label{LS1}.
\end{align}
\end{theorem}

\begin{proof}
We focus on the functions
$-s \PRenyi{1+s}(A F)_\xi$, $s \PRenyi{1-s}(F\vert A)_\xi$, and
$-s \PRenyi{1+s}(A)_\xi$, which are convex functions for $s$.
Hence, when $s$ is close to $0$,
they are approximated to
$-s H(A F)_\xi+\frac{s^2}{2} V(A F)_\xi$,
$s H( F|A)_\xi+\frac{s^2}{2} V(F|A)_\xi$,
and $-s H(A)_\xi+\frac{s^2}{2} V(A )_\xi$.
Hence,
\begin{align}
  \min \left(s \PRenyi{1+s}(A F)_\xi -s \PRenyi{1-s}(F\vert A)_\xi, s \PRenyi{1+s}(A)_\xi \right)
\end{align}
is approximated as
$s H(A)_\xi-\frac{s^2}{2}
V_\xi$, where
$V_\xi:=\max(V(A )_\xi+V(A F)_\xi,  V(F|A)_\xi ) $.
Thus,
we have
\begin{align}
&s(-n H(A)_\xi+\sqrt{n}r)+\min \big(s \PRenyi{1+s}(A F)_{\xi^{\otimes n}} -s \PRenyi{1-s}(F\vert A)_{\xi^{\otimes n}}, s \PRenyi{1+s}(A)_{\xi^{\otimes n}} \big)
\\
=&
s (-n H(A)_\xi+\sqrt{n}r)+ s n H(A)_\xi -\frac{n s^2}{2} V_\xi
+o(n s^2)
\\
=&
\sqrt{n} s r-\frac{n s^2}{2}  V_\xi +o(n s^2)\\
=&
-\frac{n}{2} V_\xi
(s- \frac{r}{\sqrt{n}V_\xi})^2+ \frac{r^2}{2 V_\xi}+o(n s^2).
\Label{Le2}
\end{align}
Since the maximum of the above value for $s$ is realized around $s=  \frac{r}{\sqrt{n}V_\xi}$, we have
\begin{align}
\lim_{n\to\infty }
{\cal L}_{\xi^{\otimes n}}(-n H(A)_\xi+\sqrt{n}r)
= \frac{r^2}{2 V_\xi},
 \Label{Le3}
\end{align}
which implies that
\begin{align}
-{\cal L}_{\xi^{\otimes n}}^{-1}(- \log \epsilon)
=n H(A)_\xi-
\sqrt{- 2n V_\xi \log \epsilon}+o(\sqrt{n}).
\Label{Le4}
\end{align}
Combining the above result with Theorem~\ref{thm:one-shot-characterization} yields~\eqref{LS1}.
\end{proof}

\subsection{Strong converse part}\Label{sec:asymptotic-converse}

In the strong converse part, we show that $\Shannon\left(\cG(\Psi_A)\right)$ is a strong converse bound
for all the quantities mentioned in Proposition~\ref{prop:enhanced-relation2}
regardless of Assumption~\ref{assp:multiplicity-free} (the multiplicity-free condition).
This concludes the first inequality of Eq.~\eqref{eq:cSiXYsXiDO}.

\begin{theorem}[Strong converse part]\Label{thm:con}
Let $\ket{\Psi}_{AF}$ be a bipartite pure quantum state. It holds that
\begin{align}
    C_{\EncPPT,\DecPPT}^\dagger\left(\Psi_{AF}\right)
&\le \Shannon\left(\cG(\Psi_A)\right),\Label{LO1}\\
    C_{\EncP,\DecSEP}^\dagger \left(\Psi_{AF}\right)
&\le \Shannon\left(\cG(\Psi_A)\right)\Label{LO2}.
\end{align}
\end{theorem}

\begin{proof}
The key tool to prove Eqs.~\eqref{LO1} and~\eqref{LO2} is the following inequality,
which is shown in \cite[(8.217)]{hayashi2016quantum}.
We denote the transpose operation on $F$ by $\tau_F$.
For a any bipartite positive semidefinite rank-one operator $X$ on $\cH_{AF}$, we have the relation
\begin{align}
|\tau_F(X)|= \sqrt{ \tr_F X} \otimes \sqrt{\tr_A X}.\Label{M6}
\end{align}
In fact, the reference \cite[(8.217)]{hayashi2016quantum} shows \eqref{M6} by using
the transpose on a specific basis.
While the map $\tau_F$ depends on the choice of the basis,
$|\tau_F(X)|$ does not depend on it as follows.
Consider the map $X \mapsto U^\dagger \tau_F( U X U^\dagger )U$ by using a unitary
on $\cH_F$.
Then, we have
\begin{align}
&|U^\dagger \tau_F( U X U^\dagger )U|^2
=
U^\dagger \tau_F( U X U^\dagger )U U^\dagger \tau_F( U X U^\dagger )U
=
U^\dagger \tau_F( U X U^\dagger ) \tau_F( U X U^\dagger )U \\
= &
U^\dagger |\tau_F( U X U^\dagger ) |^2 U
=
U^\dagger \sqrt{ \tr_F X} \otimes \sqrt{\tr_A U X U^\dagger} U \\
=&
\sqrt{ \tr_F X} \otimes (U^\dagger \sqrt{ U (\tr_A X) U^\dagger} U)
=\sqrt{ \tr_F X} \otimes \sqrt{\tr_A X}.
\end{align}
Hence, $|\tau_F(X)|$ does not depend on the choice of basis.

We now show Eq.~\eqref{LO1}. For any $\epsilon>0$, we choose a sufficiently large integer $N$ such that
any $n \ge N$ satisfies the following two conditions (i) and (ii):
\begin{enumerate}[(i)]
  \item There exists a projection $P_A$ such that
          \begin{align}
          [P_A, U_{g}]&=0 \hbox{ for }g \in G^n \Label{J1}\\
          [P_A,\cG (\Psi_A)^{\ox n}] &=0 \Label{J2}\\
          \tr P_A &\le e^{n (\Shannon(\cG (\Psi_A))+\epsilon)}, \Label{J3}\\
          \tr (I-P_A) \cG (\Psi_A)^{\ox n} &\le \epsilon ,\Label{J4} \\
          \cG(P_A ) &=P_A.\Label{J11}
          \end{align}
  \item There exists a projection $P_F$ such that
          \begin{align}
          [P_F,\Psi_F^{\ox n}] &=0,\Label{J5}\\
          \tr P_F &\le e^{n (\Shannon(\Psi_A)+\epsilon)},\Label{J6}\\
          \|P_F \Psi_F^{\ox n} P_F\| &\le e^{-n (\Shannon(\Psi_A)-\epsilon)},\Label{J7}\\
          \tr (I-P_F) \Psi_F^{\ox n} &\le \epsilon .\Label{J8}
          \end{align}
\end{enumerate}
The conditions \eqref{J1} and \eqref{J4} imply that
\begin{align}
\tr (I-P_A) \Psi_A^{\ox n} \le \epsilon .\Label{J9}
\end{align}
Hence, $[P_A,P_F]=0$ and
\begin{align}
\tr (I-P_A\ox P_F) \Psi_{AF}^{\ox n} \le 2 \epsilon .\Label{J10}
\end{align}
Let $\cC=(\{\cE^m\}, \{\Gamma^{m}\})\in(\EncPPT,\DecPPT)$ be a code for the state $\Psi_{AF}^{\ox n}$.
Since $\Gamma^m$ is a PPT operator and $\cE^m \in\EncPPT$, $(\cE^m)^*(\Gamma^m)$ is also a PPT operator, i.e.,
\begin{align}
\tau_F ((\cE^m)^* (\Gamma^m))\ge 0\Label{J14}.
\end{align}
By applying \eqref{M6} to $(I \ox P_F )\Psi_{AF}^{\ox n}(I \ox P_F)$,
the evaluation \eqref{J7} guarantees that
\begin{align}
\|\tau_F( (P_A P_F )\Psi_{AF}^{\ox n}(P_A P_F)) \|
&= \|P_A \ox I\tau_F( (I \ox P_F )\Psi_{AF}^{\ox n}(I \ox P_F)) P_A \ox I \| \\
&\leq \|\tau_F( (I \ox P_F )\Psi_{AF}^{\ox n}(I \ox P_F)) \| \\
&\leq e^{-n (\Shannon(\Psi_A)-\epsilon)}.
\end{align}
Hence, we have
\begin{align}
|\tau_F( (P_A\ox P_F )\Psi_{AF}^{\ox n}(P_A\ox P_F)) |
\le
e^{-n (\Shannon(\Psi_A)-\epsilon)} P_A\ox P_F.\Label{J12}
\end{align}
Then, we have
\begin{align}
 &\tr \Gamma^m \cE^m(
 (P_A\ox P_F )\Psi_{AF}^{\ox n}(P_A\ox P_F) ) \\
 =
 &\tr   (\cE^m)^* (\Gamma^m )
 (P_A\ox P_F )\Psi_{AF}^{\ox n}(P_A\ox P_F)  \\
 =
 & \tr  \tau_F( (\cE^m)^* (\Gamma^m ))
\tau_F( (P_A\ox P_F )\Psi_{AF}^{\ox n}(P_A\ox P_F))  \\
\stackrel{(a)}{\leq}
 &\tr  \tau_F( (\cE^m)^* (\Gamma^m ))
|\tau_F( (P_A\ox P_F )\Psi_{AF}^{\ox n}(P_A\ox P_F)) | \\
\stackrel{(b)}{\leq}
 &\tr  \tau_F( (\cE^m)^* (\Gamma^m ))
e^{-n (\Shannon(\Psi_A)-\epsilon)} P_A\ox P_F\\
\stackrel{(c)}{=}
 & \tr   (\cE^m)^* (\tau_F(\Gamma^m ))
\cG(e^{-n (\Shannon(\Psi_A)-\epsilon)} P_A\ox P_F )\\
 =
 & \tr
 \tau_F(\Gamma^m )
  \cE^m ( \cG( e^{-n (\Shannon(\Psi_A)-\epsilon)} P_A\ox P_F ))\\
 =
 & \tr
 \tau_F(\Gamma^m ) \cG( e^{-n (\Shannon(\Psi_A)-\epsilon)} P_A\ox P_F )\\
 =
 & \tr
 \tau_F(\Gamma^m ) e^{-n (\Shannon(\Psi_A)-\epsilon)}  P_A\ox P_F\\
 =
 & e^{-n (\Shannon(\Psi_A)-\epsilon)} \tr \Gamma^m   P_A\ox P_F,\Label{J13}
\end{align}
where
$(a)$, $(b)$ and $(c)$ follow from
\eqref{J14}, \eqref{J12} and \eqref{J11}, respectively.
Hence, we have
\begin{align}
 s(\cC)
 = &\frac{1}{|{\cal M}|} \sum_{m} \tr \Gamma^m \cE^m(\Psi_{AF}^{\ox n}) \\
 \stackrel{(a)}{\leq}
 &\frac{1}{|{\cal M}|} \sum_{m} \tr \Gamma^m \cE^m(
 (P_A\ox P_F )\Psi_{AF}^{\ox n}(P_A\ox P_F) ) +2\epsilon \\
 \stackrel{(b)}{\leq}
 &\frac{1}{|{\cal M}|}
 e^{-n (\Shannon(\Psi_A)-\epsilon)}  \sum_{m} \tr \Gamma^m   P_A\ox P_F
 +2\epsilon \\
 = &\frac{1}{|{\cal M}|}
 e^{-n (\Shannon(\Psi_A)-\epsilon)}  \tr P_A \ox P_F
 +2\epsilon \\
 \stackrel{(c)}{\leq}
&\frac{1}{|{\cal M}|}
 e^{-n (\Shannon(\Psi_A)-\epsilon)}  e^{n (\Shannon(\cG (\Psi_A))+\epsilon)}
 e^{n (\Shannon(\Psi_A)+\epsilon)}
 +2\epsilon \\
 \le &\frac{1}{|{\cal M}|}
  e^{n (\Shannon(\cG (\Psi_A))+3 \epsilon)}
 +2\epsilon ,
\end{align}
where
$(a)$, $(b)$, and $(c)$ follow from the pair of \eqref{J4} and \eqref{J8},
 \eqref{J13} and the pair of \eqref{J3} and \eqref{J6},  respectively.
Thus,
\begin{align}
|{\cal M}|
\le \frac{1}{ s(\cC)- 2\epsilon } e^{n (\Shannon(\cG (\Psi_A))+3\epsilon)},\Label{J15}
 \end{align}
which implies \eqref{LO1}.

Next, we show \eqref{LO2}.
Let $\cC=(\{\cE^m\}, \{\Gamma^{m}\})\in(\EncP,\DecSEP)$ be a code for the state $\Psi_{AF}^{\ox n}$.
Since $\Gamma^m$ is a separable operator and $\cE^m\in\EncP$,
$ (\cE^m)^* (\Gamma^m)$ is also a separable operator.
Hence, $ (\cE^m)^* (\Gamma^m)$ is a PPT operator, i.e., we have \eqref{J14}.
Therefore, in the same way, we can show \eqref{J15},
which implies \eqref{LO2}.
\end{proof}

\section{Proof of Theorem~\ref{thm:asymptotic empty}}\Label{S6}

\begin{proof}
Applying the channel coding theorem to the classical-quantum channel
$g \mapsto U_g \Psi_{A}U_g^\dagger$~\cite{hayashi2016quantum},
we can easily derive the following capacity formula:
\begin{align}
C_{\EncG,\DecL}(\Psi_{AF}) = \Rel\left(\Psi_{A}\| \cG(\Psi_{A})\right).
\end{align}
Hence, to show \eqref{eq:asymptotic empty} it is sufficient to show the strong converse part
\begin{align}
C^\dagger_{\EncP,\DecL}(\Psi_{AF}) \leq \Rel\left(\Psi_{A}\| \cG(\Psi_{A})\right).
\Label{NH4}
\end{align}
In almost the same way as~\eqref{eq:enhanced-relation-3}, the strong converse argument \eqref{NH4} can be shown
by invoking the meta-converse technique originally invented in~\cite{nagaoka2001strong} and
further investigated in~\cite[Chapter 3]{hayashi2005asymptotic}.
In the following Proposition~\ref{emp-strong-converse} (which will be proved shortly),
we upper bound the success probability of any one-shot code $\cC\in(\EncP,\DecL)$
in terms of the sandwiched quantum \Renyi entropy, then the strong converse bound follows by block coding.
Notice that $\lim_{\alpha\to1}\SRRel{\alpha}\left(\Psi_A\rel\cG(\Psi_A)\right)=\Rel\left(\Psi_A\rel\cG(\Psi_A)\right)$.
What's more, $\SRRel{\alpha}$ is continuous and monotonically decreasing in $\alpha$.
Applying the standard argument outlined
in~\cite{nagaoka2001strong,sharma2013fundamental}, we obtain from Proposition~\ref{prop:one-shot-strong-converse} that
$\Rel\left(\Psi_A\rel\cG(\Psi_A)\right)$ is actually a strong converse bound.
\end{proof}

\begin{proposition}\Label{emp-strong-converse}
Any dense coding code $\cC\in(\EncP,\DecL)$ obeys the following bound for arbitrary $\alpha\in (1,\infty)$:
\begin{align}\Label{eq:emp-strong-converse}
s(\cC) \leq \exp\left\{\frac{\alpha-1}{\alpha}\left(
\SRRel{\alpha}\left(\Psi_A\rel\cG(\Psi_A)\right)
 - \log\vert\cC\vert\right)\right\},
\end{align}
where $\SRRel{\alpha}$ is the sandwiched quantum \Renyi entropy defined in~\eqref{MO2}.
\end{proposition}

\begin{proof}
Given a code $\cC=(\{\cE^m\}_m,\{\Gamma^{\wh{m}}\}_{\wh{m}})\in(\EncP,\DecL)$,
we define the following two quantum states:
\begin{align}
    \rho_{MA} &:= \frac{1}{\vert\cM\vert}\sum_{m} \proj{m}_M\ox
    \cE^{m}(\Psi_A),\\
    \sigma_{MA} &:= \pi_M\ox \cG(\Psi_A),
\end{align}
where $\xi_{MXA}$ serves as a \emph{test state}.
The positive operator
\begin{align}
T:=  \sum_{m} \proj{m}_M\ox \Gamma^{m}
\end{align}
satisfies
\begin{align}
\tr T \rho_{MA}&=  \frac{1}{\vert\cM\vert}\sum_{m}p_{\wh{M}M}(m\vert m)
= s(\cC), \\
\tr T \sigma_{MA}
&=  \frac{1}{\vert\cM\vert} \tr\left[\sum_m \Gamma^{m}\cG(\Psi_A)\right]
\stackrel{(a)}{=} \frac{1}{\vert\cM\vert} \tr\left[\cG(\Psi_A)\right] = \frac{1}{\vert\cM\vert},
\end{align}
where $(a)$ follows from the fact that $\{\Gamma^{m}\}_{m}$ is a quantum measurement.
Applying the information processing inequality to the binary measurement
$\{ T,I-T \}$, we have
\begin{align}
s(\cC)^\alpha \cdot \left(\frac{1}{\vert\cC\vert}\right)^{1-\alpha}
+(1-s(\cC))^\alpha \cdot \left(1-\frac{1}{\vert\cC\vert}\right)^{1-\alpha}
&\leq e^{(\alpha-1)\PRRel{\alpha}\left(\rho_{MA}\rel\sigma_{MA}\right)}.
\end{align}
Thus, we have the following:
\begin{align}
 s(\cC)^\alpha \cdot \vert\cC\vert^{\alpha-1}
&= s(\cC)^\alpha \cdot \left(\frac{1}{\vert\cC\vert}\right)^{1-\alpha}
\leq s(\cC)^\alpha \cdot \left(\frac{1}{\vert\cC\vert}\right)^{1-\alpha}
+(1-s(\cC))^\alpha \cdot \left(1-\frac{1}{\vert\cC\vert}\right)^{1-\alpha} \\
&\leq e^{(\alpha-1)\PRRel{\alpha}\left(\rho_{MA}\rel\sigma_{MA}\right)}
\\
&=\frac{1}{\vert\cM\vert}\sum_{m}
e^{(\alpha-1)\SRRel{\alpha}\left(\cE^{m}(\Psi_A)\rel\cG(\Psi_A)\right)} \\
&\stackrel{(a)}{=}
\frac{1}{\vert\cM\vert}\sum_{m}
e^{(\alpha-1)\SRRel{\alpha}\left(\cE^{m}(\Psi_A)\rel\cE^{m} \circ\cG(\Psi_A)\right)} \\
&\stackrel{(b)}{\le}
\frac{1}{\vert\cM\vert}\sum_{m}
e^{(\alpha-1)\SRRel{\alpha}\left(\Psi_A\rel\cG(\Psi_A)\right)} ,
\end{align}
where $(a)$ follows from the condition $\cE^{m} \circ\cG=\cG$ for arbitrary $\cE^m\in\EncP$
(cf. the definition in Eq.~\eqref{eq:free-operations3})
and $(b)$ follows from the information processing inequality of
the sandwiched quantum \Renyi entropy for positive maps~\cite[Theorem 2]{muller2017monotonicity}.
\end{proof}

\section{Proof of Theorems~\ref{thm:w-con2} and~\ref{thm:asymptotic-direct2}}\label{appx:w-con2}

\begin{proofwithpara}[Proof of Theorem~\ref{thm:w-con2}]
Notice that Eq.~\eqref{eq:w-con2-1} can be shown in the same way as \eqref{eq:enhanced-relation-3}.

Now we show Eq.~\eqref{eq:w-con2-2}.
For a state $\rho_{AF}$, we have a trace-preserving positive operation
$\cE_F$ on $\cH_F$
such that $\rho_{AF}= \cE_F(\Psi_{AF})$, where $ \Psi_{AF}$
is a purification of $\rho_A$.
Let $\cC=(\{\cE^m\}, \{\Gamma^{m}\})\in(\EncP,\DecSEP)$ be a code for the state $\rho_{AF}^{\ox n}$.
Since $\{ \cE_F^*(\Gamma^{m})\}$ is a separable measurement, where $\cE_F^*$ is the dual map of $\cE_F$,
we define a code $\hat{\cC}=(\{\cE^m\}, \{ \cE_F^*(\Gamma^{m})\})$
in the encoder-decoder pair $\in(\EncP,\DecSEP)$ for the state $\Psi_{AF}^{\ox n}$.
Since $s(\cC)=s(\hat{\cC})$, \eqref{eq:w-con2-2} follows from \eqref{LO2}.

Next, we show \eqref{eq:w-con2-3}.
For a state $\rho_{AF}'$, we have a trace-preserving operation $\cE_F'\in \channel{F\to F}_{\rm ppt}$
such that $\rho_{AF}= \cE_F'(\Psi_{AF})$.
Let $\cC=(\{\cE^m\}, \{\Gamma^{m}\})\in(\EncPPT,\DecPPT)$ be a code for the state ${\rho_{AF}'}^{\ox n}$.
Since $\{ {\cE_F'}^*(\Gamma^{m})\}$ is a separable measurement,
we define a code $\hat{\cC}'=(\{\cE^m\}, \{ {\cE_F'}^*(\Gamma^{m})\})$
in the encoder-decoder pair $\in(\EncPPT,\DecPPT)$ for the state $\Psi_{AF}^{\ox n}$.
Since $s(\cC)=s(\hat{\cC}')$, \eqref{eq:w-con2-3} follows from \eqref{LO1}.
\end{proofwithpara}

\vspace*{0.1in}

\begin{proofwithpara}[Proof of Theorem~\ref{thm:asymptotic-direct2}]
First, we show Eq.~\eqref{LS12}.
Let $\cC=(\{\cE^m\},\{ \Lambda^x\},\{\Gamma^{\wh{m}\vert x}\})$
be a code in the encoder-decoder pair $(\EncG,\DecOne)$ for the state $\Psi_{AF}^{\ox n}$.
Since the operation $\cE_F$ is a trace-preserving positive operation,
the operation ${\cE_F}^*$ is a unit-preserving positive operation.
Since $\{ {\cE_F}^*(\Lambda^x)\}$ is a POVM on $\cH_F^{\ox n}$,
we define a code $\hat{\cC}=
(\{\cE^m\},\{ {\cE_F}^*(\Lambda^x)\},\{\Gamma^{\wh{m}\vert x}\})$
in the encoder-decoder pair $(\EncG,\DecOne)$ for the state $\rho_{AF}^{\ox n}$.
Since $s(\cC)=s(\hat{\cC})$, \eqref{LS12} follows from \eqref{LS1}.

Next, we show Eq.~\eqref{LS13}.
Let $\cC=(\{{\cE^m}'\},\{ {\Lambda^x}'\},\{{\Gamma^{\wh{m}\vert x}}'\})$
be a code in the encoder-decoder pair $(\EncG,\DecOne)$ for the state ${\Psi_{AF}'}^{\ox n}$.
Since $\{ {{\cE_F'}}^*(\Lambda^x)\}$ is a POVM on $\cH_F^{\ox n}$,
we define a code $\hat{\cC}'=
(\{{\cE^m}'\},\{ {{\cE_F'}}^*({\Lambda^x}')\},
\{{\Gamma^{\wh{m}\vert x}}'\})$
in the encoder-decoder pair $(\EncG,\DecOne)$ for the state
${\rho_{AF}'}^{\ox n}$.
Since $s(\cC)=s(\hat{\cC}')$, \eqref{LS13} follows from \eqref{LS1}.
\end{proofwithpara}

\end{document}